\documentclass[12pt,reqno]{article}
\usepackage{amssymb,latexsym}
\usepackage{graphicx,amsmath,amssymb}
\usepackage{color}
\usepackage{amsthm}
\usepackage{stmaryrd}
\usepackage[dvips,pdfstartview=FitH,backref,pdfpagemode=None,colorlinks=true,citecolor=darkblue,linkcolor=red]{hyperref} 
\usepackage{mathrsfs}
\usepackage{caption}
\usepackage{fullpage}
\usepackage{phuong}
\usepackage[disable,colorinlistoftodos]{todonotes} 
\usepackage[usenames,dvipsnames]{pstricks}
\newlength{\defbaselineskip}
\newcommand{\doublespacing}{\setlength{\baselineskip}{1.1\defbaselineskip}}
\newcommand{\Comment}[1]{}

\newcounter{tmpt}

\newtheorem{theorem}{Theorem}
\newtheorem*{theorem*}{Theorem}
\newtheorem{lemma}[theorem]{Lemma}

\newtheorem{definition}{Definition}
\newtheorem{proposition}[theorem]{Proposition}
\newtheorem{corollary}[theorem]{Corollary}
\newtheorem{fact}[theorem]{Fact}

\newtheorem*{remark*}{Remark}

\newtheorem*{notation*}{Notation}

\newtheorem*{claim*}{Claim}
\newtheorem*{proposition*}{Proposition}
\newtheorem*{corollary*}{Corollary}

\newtheorem*{A-LThm}{Amitsur-Levitzki Theorem}
\newtheorem*{main-open}{Conjecture II}
\newtheorem*{main-lower-bound}{Theorem \ref{thm:main_lower_bound}}

\renewcommand{\dot}[1]{{#1}^\star}

\def\squareforqed{\hbox{\rlap{$\sqcap$}$\sqcup$}}
\def\qed{\ifmmode\squareforqed\else{\unskip\nobreak\hfil
\penalty50\hskip1em\null\nobreak\hfil{\tt QED}
\parfillskip=0pt\finalhyphendemerits=0\endgraf}\fi}

\newcommand\F{\ensuremath{\mathbb F}}
\newcommand\N{\ensuremath{\mathbb N}}
\newcommand\Z{\ensuremath{\mathbb Z}}

\newcommand{\ignore}[1]{}
\newcommand\PP{{\mathbb P}}

\newcommand\PC{\ensuremath{\PP_c}}
\newcommand\PMd{\ensuremath{\PP_{{\rm Mat}_d}(\F)}}
\newcommand\PMtwo{\ensuremath{\PP_{{\rm Mat}_2}(\F)}}

\newcommand{\poly}{\hbox{poly}}
\newcommand{\cd}{\cdot}
\renewcommand{\l}{\ell}

\newcommand{\case}[1]{\ind\textbf{Case #1}:\,}

\newcommand {\ind} {\noindent}

\newcommand {\para}[1] {\paragraph{#1}}

\DeclareMathAlphabet{\mathitbf}{OML}{cmm}{b}{it}

\newcommand{\matd}{{\ensuremath{{\rm Mat}_d(\F)}}}
\newcommand{\mattwo}{{\ensuremath{{\rm Mat}_2(\F)}}}
\newcommand{\matone}{{\ensuremath{{\rm Mat}_1(\F)}}}
\newcommand{\matthree}{{\ensuremath{{\rm Mat}_3(\F)}}}
\newcommand{\freea}{\ensuremath{\F\langle X\rangle}}
\newcommand{\dbyd}{\ensuremath{d\times d}}

\newenvironment{examples}{\QuadSpace\par\noindent{\bf Examples}:}{\HalfSpace}

\newenvironment{comment}{\QuadSpace\par\noindent{\bf Comment}:}{\HalfSpace}

\ifx\proof\undefined
\newenvironment{proof}{\QuadSpace\par\noindent{\bf Proof}:}{\EndProof\HalfSpace}
\fi

\newenvironment{proofclaim}{\QuadSpace\par\noindent{\bf Proof of claim}:}
{\vrule width 1ex height 1ex depth 0pt $_{\textrm{\,Claim}}$ \HalfSpace}

\newcommand{\QuadSpace}{\vspace{0.25\baselineskip}}
\newcommand{\HalfSpace}{\vspace{0.5\baselineskip}}

\newcommand{\EndProof}{ \hfill \vrule width 1ex height 1ex depth 0pt }

\definecolor{bluetxt}{rgb}{0,0,.5}
\definecolor{myred}{rgb}{0.6,0.0,0.1}
\definecolor{greentxt}{rgb}{0,.5,0}
\definecolor{redtxt}{rgb}{0.1,0.1,0.65}
\definecolor{purpletxt}{rgb}{0.6,0.1,0.7}
\definecolor{black}{rgb}{.0,.0,.0}
\definecolor{verydarkblue}{rgb}{.0,.0,.4}
\definecolor{darkblue}{rgb}{.0,.0,.4}
\definecolor{lightgray}{rgb}{.7,.7,.7}

\ifx\proof\undefined

 \fi

\newcommand{\iddotodo}[2][]
{\todo[size=\tiny, noline, caption={#2}, #1, linecolor=green!70!white,         backgroundcolor=blue!10!white,bordercolor=white]
{{#2}}}
\newcommand{\iddocomment}[2][]
{\todo[size=\tiny, caption={#2}, #1, linecolor=green!70!white,         backgroundcolor=blue!10!white,bordercolor=white]
{{#2}}}

\newcommand{\iddofix}[2][]
{\todo[inline, size=\footnotesize, caption={#2}, #1, linecolor=green!70!white,         backgroundcolor=red!30!white,bordercolor=white]
{{#2}}}

\newcommand{\iddofixl}[2][]
{\todo[size=\tiny, caption={#2}, #1, linecolor=green!70!white,         backgroundcolor=red!30!white,bordercolor=white]
{{#2}}}

\newcommand{\set}[1]{\left\{#1\right\}}

\newcommand{\nx}[1]{#1_1,\ldots,#1_{n}}

\renewcommand{\a}{a^{(i)}}
\renewcommand{\t}[1]{\overline{#1}}
\newcommand{\anbra}[1]{\ensuremath{\left[ #1\right]}}
\newcommand{\degr}[2]{\left(#1\right)^{(#2)}}
\newcommand{\s}{\sigma}
\newcommand{\zd}[2]{(#1)_Z^{#2}}
\newcommand{\Number}{(2d+1)\cd\l}

\renewcommand{\S}{\mathcal{S}}

\newcommand{\convert}[1]{\llbracket#1\rrbracket_d}

\newcommand{\ideal}[1]{\ensuremath{\left\langle #1\right\rangle}}
\newcommand\ii{{\cal I}}
\newcommand{\U}{\mathcal U_{ij}}
\newcommand{\V}{\mathcal V_{ij}}
\newcommand{\uzv}{\sum_{i=1}^n\sum_{j}\U z_i\V}

\allowdisplaybreaks

\author{Fu Li\thanks{Institute for Theoretical Computer Science, The Institute for Interdisciplinary Information Sciences (IIIS), Tsinghua University, Beijing.  Supported in part by the National Basic Research Program of China Grant 2011CBA00300, 2011CBA00301, the National Natural Science Foundation of China Grant 61033001, 61361136003 and NSFC grant 61373002.} \and Iddo Tzameret\thanks{Department of Computer Science, Royal Holloway, University of London. Email:  \texttt{iddo.tzameret@gmail.com} ~~Supported in part by the NSFC Grant 61373002.}}

\usepackage[normalem]{ulem}
\usepackage{pifont}
\usepackage{boxedminipage}

\begin{document}

\title{Generating Matrix Identities and Proof Complexity}
\maketitle
\doublespacing
\thispagestyle{empty}
\begin{abstract}
Motivated by the fundamental lower bounds questions in proof complexity, we initiate the study of matrix identities as hard instances for \emph{strong} proof systems.
A \emph{matrix identity} of \dbyd\ matrices over a field \F, is a non-commutative polynomial \(f(x_1,\ldots,x_n)\) over \F\ such that $f$ vanishes on every \dbyd\ matrix assignment to its variables.

We focus on \textit{arithmetic proofs}, which are proofs of polynomial identities operating with arithmetic circuits and whose axioms are the polynomial-ring axioms (these proofs serve as an algebraic analogue of the Extended Frege propositional proof system; and over $GF(2)$ they   constitute formally a sub-system of Extended Frege \cite{HT12}). We introduce a decreasing in strength hierarchy of proof systems within arithmetic proofs, in which the $d$th level is a sound and complete proof system for proving \dbyd\ matrix identities (over a given field). For each level $d>2$ in the hierarchy, we establish a proof-size lower bound in terms of the number of variables in the matrix identity proved: we show the existence of a family of matrix identities $f_n$ with $n$ variables, such that any proof of $f_n=0$ requires $\Omega(n^{2d})$ number of lines.

The lower bound argument  uses fundamental results  from the theory of algebras with polynomial identities together with a generalization of the arguments in \cite{Hru11}. Specifically, we establish an unconditional lower bound on the minimal number of generators needed to generate a matrix identity, where the generators are substitution instances of elements from any given finite basis of the matrix identities; a result that might be of independent interest.

We then set out to study matrix identities as hard instances for (\textit{full}) arithmetic proofs. We present two conjectures, one about non-commutative arithmetic circuit complexity and the other about proof complexity, under which up to \textit{exponential-size} lower bounds on arithmetic proofs (in terms of the arithmetic circuit size of the identities proved) hold. Finally, we discuss the applicability of our approach to strong \textit{propositional} proof systems such as Extended Frege.

\end{abstract}


\iddotodo{Write about unity in associative algebras.}
\iddocomment{Explain the adjective strong proof system}

\newpage
\pagenumbering{arabic}
\section{Background}
Proving super-polynomial size lower bounds on strong propositional proof systems, like the Extended Frege system, is a major open problem in proof complexity, and in general is among a handful of fundamental hardness questions in computational complexity theory. An Extended Frege proof is simply a textbook logical proof system for establishing Boolean tautologies, in which one starts from basic tautological axioms written as Boolean formulas,  and derives, step by step, new tautological formulas from previous ones by using a finite set of logical sound derivation rules; including  the so-called \textit{extension axiom} enabling  one to denote a possibly big formula by a \textit{single }new variable (where the variable is used neither before in the proof nor in the last line of the proof). It is not hard to show (see \cite{Jer04}) that Extended Frege can equivalently be defined as a logical proof system operating with Boolean \emph{circuits} (and without the extension axiom\footnote{An additional simple technical axiom is needed to formally define this proof system (\cite{Jer04}).}).

Lower bounds on Extended Frege proofs can be viewed as lower bounds on certain nondeterministic algorithms for establishing the unsatisfiability of Boolean formulas (and thus as a progress towards separating \NP\ from \coNP). It is also usually considered (somewhat informally) as related to establishing (explicit) Boolean circuit size lower bounds. In fact,  it has also another highly significant consequence, that places such a lower bound as a  step towards separating \Ptime\ from \NP: showing any super-polynomial lower bound on the size of Extended Frege proofs implies that, at least with respect to ``polynomial-time reasoning" (namely, reasoning in the formal theory of arithmetic denoted $S^1_2$), it is not possible to prove that $\Ptime=\NP$; or in other words, it is consistent with $S^1_2$ that $\Ptime\neq$\NP\ (cf.~\cite{KP89}).

Accordingly, proving Extended Frege lower bounds is considered a very hard problem. In fact, even \emph{conditional} lower bounds on strong proof systems, including  Extended Frege, are not known and are considered very interesting;\footnotemark \ here, we mean  a condition that is  different from $\NP\neq\coNP$ (see \cite{Pud08};  the latter condition immediately  implies that any propositional proof system admits a family of tautologies with no polynomial-size proofs \cite{CR79}). The only size lower bound on Extended Frege proofs that is known to date is linear $\Omega(n)$ (where $n$ is the size of the tautological formula proved; see \cite{Kra95} for a proof).
Establishing \textit{super-linear }size lower bounds on Extended Frege proofs is thus a highly interesting open problem.

That said, although proving Extended Frege lower bounds is a fundamental open problem in complexity, it is quite unclear whether such lower bounds are indeed far from reach or beyond current techniques (in contrast to other fundamental hardness problems in complexity, such as strong explicit Boolean circuit lower bounds, for which formal so-called barriers are known).

\footnotetext{Informally, we call a proof system \emph{strong} if  there are no known (non-trivial) size lower bounds on proofs in the system and further such lower bounds are believed to be outside the realm of current techniques.}

Another feature of proof complexity is that, in contrast to circuit  complexity, even the \textit{existence of non-explicit }hard instances for strong propositional proof systems, including Extended Frege, are unknown. For instance, simple counting arguments cannot establish super-linear size lower bounds on Extended Frege proofs (in contrast to Shannon's counting argument which gives non-explicit lower bounds on circuit size, but does not in itself yield complexity class separations). Thus, the existence of non-explicit hard instances
in proof complexity is sufficient for the purpose of lower
bounding the size of strong proof systems.

Furthermore, for strong proof systems there are almost no hard candidates, namely, tautologies that are believed to require long proofs in these systems (see Bonet, Buss and Pitassi \cite{BBP95}); except, perhaps for random $k$-CNF formulas near the satisfiability threshold. But for the latter instances, even lower bounds on Frege proofs of constant-depth are unknown. It is worth noting also that Razborov \cite{Razb03} and especially Kraj\'{i}\v{c}ek (see e.g., \cite{Kra10-forcing}) had proposed some tautologies as hard candidates for strong proof systems.

Due to the lack of progress on establishing lower bounds on strong propositional proof systems, it is interesting, and potentially helpful,  to turn our eyes to an \textit{algebraic analogue} of strong propositional proof systems, and try first to prove nontrivial size lower bounds in such settings. Quite recently, such algebraic analogues of Extended Frege (and Frege, which is Extended Frege without the extension axiom) were investigated by Hrube\v s and the second author \cite{HT08,HT12}. These proof systems denoted \PC(\F), called simply \emph{arithmetic proofs}, operate with algebraic equations of the form $F=G$, where $F$ and $G$ are algebraic circuits  over a given field \F. An arithmetic  proof  of a polynomial identity is a sequence of identities between algebraic circuits derived by means of simple syntactic manipulation representing the polynomial-ring axioms (e.g., associativity, distributivity, unit element, field identities, etc.; see Definition \ref{def:arithmetic_proofs}). Although arithmetic proof systems are not propositional proof systems, namely they do not prove propositional tautologies, they can be regarded nevertheless as \textit{fragments} of the propositional Extended Frege proof system when the field considered is $GF(2)$. That is, every arithmetic proof over $GF(2)$ of a polynomial identity  (considered as a propositional  tautology) can formally be viewed also as an Extended Frege proof.\footnote{In fact, it is probably true (but was not formally verified) that arithmetic proofs are fragments of propositional proofs also over any other finite field, as well as over the ring of integers (when restricted to up to exponentially big integers). That is, it is probably true that every polynomial identity proved with an arithmetic proof over the given field or ring, can be proved with at most a polynomial increase in
size in Extended Frege when we fix a certain natural translation between polynomial identities over the field or ring and propositional tautologies. The reason for this is that one could plausibly polynomially simulate arithmetic proofs over such fields or rings with  propositional proofs in which numbers are encoded as bit-strings.}

Apart from the hope that arithmetic proofs would shed light on propositional proof systems, the study of arithmetic proofs is motivated by the Polynomial Identity Testing (PIT) problem, namely the problem of deciding if a given algebraic circuit computes the zero polynomial. As a decision problem, polynomial identity testing can be solved by an efficient randomized algorithm \cite{Sch80,Zip79}, but no efficient deterministic algorithm is known. In fact, it is not even known whether there is a polynomial time non-deterministic algorithm or, equivalently, whether  PIT is in \NP. An arithmetic  proof system can thus be interpreted as a specific non-deterministic algorithm for PIT: in order to verify that an arithmetic circuit $C$ computes the zero polynomial, it is sufficient to guess an arithmetic  proof of $C=0$.  Hence, if every true equality has a polynomial-size proof then PIT is in \NP. Conversely, the arithmetic proof system captures the common syntactic procedures used to establish equality between algebraic expressions. Thus, showing the existence of identities that require super-polynomial arithmetic proofs would imply that those syntactic procedures are not enough to solve the PIT problem efficiently.\footnote{It is worth emphasizing again that arithmetic proofs are different than algebraic \textit{propositional} proof systems like the Polynomial Calculus \cite{CEI96} and related systems. The latter prove propositional tautologies (a \textit{\textbf{coNP}} language) while the former prove formal polynomial identities written as equations between algebraic circuits (a \textit{\textbf{coRP}} language).}

The emphasis in \cite{HT08,HT12} was mainly on demonstrating non-trivial  \textit{upper bounds} for arithmetic proofs (as well as lower bounds in very restricted settings). Since arithmetic proofs (at least over $GF(2)$), can also be considered  as propositional proofs,  arithmetic proofs were found very useful in establishing short propositional proofs for the determinant identities and other statements from linear algebra \cite{HT12}. As for \textit{lower bounds }on arithmetic proofs (operating with arithmetic circuits),  the same basic linear size lower bound known for Extended Frege \cite{Kra95} can be shown to hold for $\PC$. But any super-linear size lower bound, explicit or not, on \PC(\F) proof size (for any field \F) is open. In  \cite{HT08} it was argued that proving lower bounds even on very restricted fragments of arithmetic proofs is a highly nontrivial open problem.

The state of affairs we have described up to now shows how little is known about strong propositional (and arithmetic) proof systems, and why it is highly interesting to introduce and develop novel approaches for lower bounding proofs such as arithmetic proofs, even if these approaches yield only conditional and possibly non-explicit lower bounds; and further, to propose new kinds of hard candidates for strong proof systems.   

\section{Overview of our results}\label{sec:ovrv_of_our_results}
In this work we initiate the study of matrix identities as hard instances for strong proof systems in various settings and under different assumptions. 
The term \textit{strong} here stands for proof systems that operate with (Boolean or arithmetic) \textit{circuits}, for which we do not know any non trivial lower bound (see Sec.~\ref{sec:app_arithmetic_circuit} for the definitions of arithmetic circuits and non-commutative arithmetic circuits).

The  ultimate goal of our suggested approach is proving Extended Frege lower bounds; however, in this work we focus for most part on the seemingly (and relatively) easier task of proving arithmetic proofs \PC(\F) lower bounds, namely lower bounds on arithmetic proofs establishing polynomial identities between arithmetic circuits over a field \F.


 We introduce a new decreasing hierarchy of proof systems establishing matrix identities of a given dimension,  within arithmetic proofs (and whose first level coincides with arithmetic proofs). We obtain unconditional (polynomial) lower bounds on proof systems for matrix identities in terms of the number of variables in the identities proved. We then present two natural conjectures from arithmetic circuit complexity and proof complexity, respectively, based on which one can obtain up to exponential-size lower bounds on arithmetic proofs \PC(\F) in terms of the size of the identities proved.  \smallskip

We start by explaining what  matrix identities are, as well as providing some necessary  background from algebra.

\subsection{Matrix identities}\label{sec:ovrv:mat_identities}

For a field $\F$ let $A$ be a non-commutative (associative) \F-algebra; e.g., the algebra \matd\ of $d\times d$ matrices over  $\F$.
(Formally, $A$ is an $\F$-algebra, if $A$ is a vector space over \F\ together with a distributive multiplication operation; where multiplication in $A$ is associative (but it need not be commutative) and there exists a multiplicative unity in $A$.)

We shall always assume, unless explicitly stated otherwise, that the field \F\ has characteristic 0.

A \textit{\textbf{non-commutative polynomial}} over the field \F\ and with the variables $X:=\{x_1,x_2,\ldots\}$ is a formal sum of monomials where the product of variables is non-commuting.  Since most polynomials in this work are non-commutative \textbf{when we talk about \textit{polynomials }we shall mean \textit{non-commutative polynomials, }}unless otherwise stated.
The set of (non-commutative) polynomials with variables $X$ and over the field \F\ is denoted \freea.

We say that $f$ is a \textit{\textbf{matrix identity} of \matd} simply whenever $f$ is a  non-commutative polynomial (with coefficients from \F) that is equal to zero under any assignment of matrices from \matd\ to its variables. In other words, the  polynomial \(f(x_1,\ldots,x_n)\) over \F\ is \textit{an identity of  the algebra $A$} (and specifically, the matrix algebra \matd), if for all $\overline c\in A^n$, $f(\overline c)=0$.

%

\subsection{Stratification}\label{sec:ovrv:stratification}

A matrix identity is a non-commutative polynomial vanishing over all assignments of matrices. If we consider the ``matrix'' algebra of $1 \times 1$ matrices \matone, its set of identities consists of all the non-commutative polynomials that vanish over field elements. Since, by definition, the field is commutative, the identities of \matone\ are all non-commutative polynomials such that when the product is considered as \emph{commutative} we obtain the zero polynomial; in other words, \textit{we can consider the identities of \matone\ as the set of (standard, i.e., commutative) polynomial identities}. Further, in our application we shall write all polynomials as non-commutative arithmetic circuits, and since a non-commutative arithmetic circuit is equivalent to a (commutative) arithmetic circuit (except that product gates have order on their children) \textit{we can consider the set of identities of \matone\ written as non-commutative circuits, as the set of (commutative) polynomial identities written as (commutative) arithmetic circuits.}

Using matrix identities of increasing dimensions $d$ we obtain a stratification of the language of (commutative) polynomial identities. Namely, we obtain the following strictly decreasing (with respect to containment) chain of identities:
\begin{align}\notag
\hbox{(commutative) polynomial identities} & \supsetneq \mattwo\hbox{-identities} \\ \notag
        &\supsetneq \hbox{\matthree-identities} \\ \label{eq:chain}
        &\supsetneq \ldots \\ \notag
        &\supsetneq \matd \\ \notag
        &\supsetneq {\rm Mat}_{d+1}(\F) \supsetneq\ldots\notag
\end{align}

The fact that the identities of  ${\rm Mat}_{d+1}(\F)$ are also identities of \matd\ is easy to show. The fact that the chain above is  \textit{strictly} decreasing  can be proved  either by a elementary methods \cite{Jer14-personal_communication} or as a corollary of \cite{AL50}.

\subsection{Corresponding proof systems and the main lower bound }\label{sec:ovrv:proof_system_for_matrix_identities}

We now introduce a novel  hierarchy of proof systems within arithmetic proofs \PC(\F).
For this we need the concept of a \textit{basis} of a set of identities of a given \F-algebra $A$ (e.g., the matrix algebra \matd) .

\para{Basis.}We say that a set of non-commutative polynomials $\mathcal B$ forms a \emph{\textbf{basis}} for the identities of $A$, in the following sense: for every identity $f$ of $A$ there exist non-commutative polynomials $g_1,...,g_k$, for some $k$, that are \textit{substitution instances} of polynomials from $\mathcal B$, such that  $f$ is in the two-sided ideal $\langle g_1,...,g_k \rangle$ (a \textit{\textbf{substitution instance} }of a polynomial $g(x_1,\ldots, x_n)\in\freea $ is a polynomial $g(h_1,\ldots,h_n)$, for some $h_i\in\freea$, $i\in[n]$).
\medskip

Recall that \textbf{\textit{arithmetic proofs}} \PC(\F) (see Definition \ref{def:arithmetic_proofs}) are proofs that start from basic axioms like associativity, commutativity of addition and product, distributivity of product over addition, unit element axioms, etc., in which  we derive new equations between arithmetic circuits $F=G$ using rules for adding and multiplying two previous identities. Arithmetic proofs are sound and complete proof systems for the set of (commutative) polynomial identities, written as equations between arithmetic circuits.

Notice that if one takes out the Commutativity Axiom $f\cd g = g\cd f$ from arithmetic proofs, we get a proof system for establishing non-commutative polynomial identities written as non-commutative arithmetic circuits (we can assume that product gates appearing in arithmetic proofs have order on their children).

\para{The proof systems \PMd.}
For any field \F\ (of characteristic 0), any $d\ge 1,$ and any basis \(\mathcal B\) of the identities of \matd, we define the following proof system \PMd, which is sound and complete for the identities of \matd\ (written as equations of non-commutative circuits): consider the proof systems  \PC(\F)\ (Definition \ref{def:arithmetic_proofs}) and \textit{replace }the commutativity axiom \(h\cd g=g\cd h\) by  a finite basis \(\mathcal B\) of the identities of  \matd\  (namely, add a new axiom $H=0$ for each polynomial $h$ in the basis, where $H$ is a non-commutative algebraic circuit computing $h$).\footnote{Formally, we should fix a specific finite basis \(\mathcal B\) for the sake of definiteness of \PMd. However, different choices of bases can only increase the number of lines in a \PMd-proof by a constant factor.} Additionally, add the axioms of  distributivity of product over addition from \textit{both} left and right (this is needed because we do not have anymore the commutativity axiom in our system to simulate both distributivity axioms).

Note that $ \PC(\F)$ can be considered as ${\PP_{{\rm Mat}_1}(\F)}$, since the commutator $[g,h]$ is an axiom of  \PC(\F) and the commutator is a basis of the identities of \matone.

\medskip

%
\begin{center}
\scalebox{.35} 
{
\begin{pspicture}(0,-6.62)(19.47625,6.62)
\definecolor{color1043b}{rgb}{0.9176470588235294,0.9176470588235294,0.9176470588235294}
\definecolor{color1045b}{rgb}{0.8431372549019608,0.8431372549019608,0.8431372549019608}
\definecolor{color1047b}{rgb}{0.8,0.8,0.8}
\definecolor{color1049b}{rgb}{0.7568627450980392,0.7568627450980392,0.7568627450980392}
\definecolor{color1066b}{rgb}{0.6666666666666666,0.6666666666666666,0.6666666666666666}
\definecolor{color1066}{rgb}{0.4,0.4,0.4}
\definecolor{color1068b}{rgb}{0.6274509803921569,0.6274509803921569,0.6274509803921569}
\pscircle[linewidth=0.04,dimen=outer,fillstyle=solid,fillcolor=color1043b](12.85625,0.0){6.62}
\pscircle[linewidth=0.04,dimen=outer,fillstyle=solid,fillcolor=color1045b](12.92625,0.01){4.83}
\pscircle[linewidth=0.04,dimen=outer,fillstyle=solid,fillcolor=color1047b](12.86625,-0.01){3.01}
\pscircle[linewidth=0.04,dimen=outer,fillstyle=solid,fillcolor=color1049b](12.93625,-0.04){1.82}
\usefont{T1}{ptm}{m}{n}
\rput(12.912656,5.73){(Commutative) Polynomial Identities}
\usefont{T1}{ptm}{m}{n}
\rput(12.757031,3.63){Mat$_{2}(\mathbb{F})$-identities}
\usefont{T1}{ptm}{m}{n}
\rput(13.001094,5.23){over $\mathbb{F}$}
\usefont{T1}{ptm}{m}{n}
\rput(12.917031,2.07){Mat$_{3}(\mathbb{F})$-identities}
\usefont{T1}{ptm}{m}{n}
\rput(12.897031,0.99){Mat$_{4}(\mathbb{F})$-identities}
\pscircle[linewidth=0.03,linecolor=color1066,dimen=outer,fillstyle=solid,fillcolor=color1066b](12.94625,-0.05){0.81}
\pscircle[linewidth=0.025999999,linecolor=color1066,dimen=outer,fillstyle=solid,fillcolor=color1068b](12.96625,-0.03){0.53}
\psdots[dotsize=0.04](13.01625,0.14)
\psdots[dotsize=0.04](13.03625,-0.06)
\psdots[dotsize=0.04](13.01625,-0.28)
\psline[linewidth=0.036cm,linecolor=color1066,arrowsize=0.05291667cm 2.0,arrowlength=1.4,arrowinset=0.4]{->}(5.71625,5.14)(9.31625,5.16)
\usefont{T1}{ptm}{m}{n}
\rput(2.7123437,5.18){\Large Arithmetic proofs $\PC(\F)$}
\psline[linewidth=0.036cm,linecolor=color1066,arrowsize=0.05291667cm 2.0,arrowlength=1.4,arrowinset=0.4]{->}(5.69625,1.94)(9.23625,1.96)
\usefont{T1}{ptm}{m}{n}
\rput(4.3653126,2.0){\LARGE $\PMtwo$}
\psline[linewidth=0.036cm,linecolor=color1066,arrowsize=0.05291667cm 2.0,arrowlength=1.4,arrowinset=0.4]{->}(5.77625,0.76)(10.63625,0.8)
\usefont{T1}{ptm}{m}{n}
\rput(4.4753127,0.8){\LARGE ${\PP_{{\rm Mat}_3}(\F)}$}
\usefont{T1}{ptm}{m}{n}
\rput(4.4953127,-0.16){\LARGE ${\PP_{{\rm Mat}_4}(\F)}$}
\psline[linewidth=0.036cm,linecolor=color1066,arrowsize=0.05291667cm 2.0,arrowlength=1.4,arrowinset=0.4]{->}(5.75625,-0.26)(11.41625,-0.16)
\end{pspicture}
}


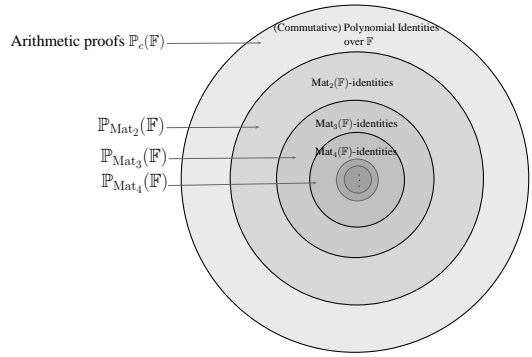
\captionof{figure}{\footnotesize{Illustration of the stratification of the language of polynomial identities and the corresponding proof systems for each language.}}
\end{center}

Our main result is an unconditional lower bound on the size (in fact the number of lines\footnote{A \textit{proof-line} is any equation $F=G$   between arithmetic circuits appearing in the proof.}) of $\PP_{{\rm Mat}_d}(\F) $ proofs, for any $d$, \textit{in terms of the number of variables $n$} in the matrix identity proved:

\begin{theorem}[Main lower bound]
\label{thm:main_lb_on_matrix_proofs}
Let \F\ be any field of characteristic 0.
For any natural number $d>2$ and  every finite basis \(\mathcal B\) of the  identities of \matd, there exists an identity \(f\) over \matd\ of degree $2d+1$ with $n$ variables, such that any \PMd-proof of $f$ requires $\Omega(n^{2d})$ lines.
\end{theorem}

The proof of the main lower bound---which is the main technical contribution of our work---is explained in the following subsection, and is based on a complexity measure defined on matrix identities and their generation in a (two-sided) ideal. The complexity measure is interesting by itself, and can be applied to identities of any algebra with polynomial identities (PI-algebras; see \cite{Row80,Dre99} for the theory of PI-algebras), and not only matrix identities.

\para{Comments.}

(i) When $d=2$, our proof, showing the lower bound for
\textit{every} basis \(\mathcal B\) of the identities of \mattwo, does \textit{not} hold (see Sec.~\ref{sec:conc-for-any-basis-of-matd} for an explanation).

(ii) The hard instance in the main lower bound theorem is \textit{non-explicit}.
Thus, we do not know if there are small non-commutative circuits computing the hard instances. This is the reason the lower bound holds only with respect to the number of variables $n$ in the hard-instances and not with respect to its circuit size---the latter is the more desired result in proof complexity.
Section \ref{sec:intro_hard_identities} sets out an approach to achieve this latter result.

(iii) The proof-systems ${\PP_{{\rm Mat}_d}(\F)}$
are defined using a finite basis of the identities of \matd. A very interesting feature of our lower bound argument is that it is in fact an open problem to find explicit finite
bases for the identities of \matd\ (for $d>2$; see the next sub-Section \ref{sec:ovrv:proof_of_main_lower_bound} on this).

(iv) We do not know if the hierarchy of proof systems \PMd\ for increasing $d$'s is a \emph{strictly} decreasing hierarchy (since we do not know if $\PP_{{\rm Mat}_{d-1}}(\F)$ has any speed-up over \PMd\  for identities of \matd).
\medskip

In the following subsection we give a detailed overview of the lower bound argument.

\subsubsection{Proving the main lower bound: generative complexity lower bounds}\label{sec:ovrv:proof_of_main_lower_bound}
Here we explain in details the complexity measure we define and how we obtain the lower bound on this measure. It is simple to show that our complexity measure is a lower bound on the minimal number of lines in a corresponding \PMd-proof (for the case $d=1$ this was observed in \cite{Hru11}).
\para{The complexity measure.}

Given an \F-algebra $A$ (e.g., \matd) and an identity $f$ of $A$, define
$$Q_{\mathcal B}(f)$$  as the minimal number $k$ such that there exist $g_1,\ldots,g_k\in \freea$ for which  $f\in \langle g_1,...,g_k \rangle$, and every $g_i$ is a substitution instance of some polynomial  from $\mathcal B$. (Note that each substitution instance, even of the same polynomial from $\mathcal B $, adds to $Q_{\mathcal B}(f)$.) We sometimes call $Q_{\mathcal B}(f)$ \textit{\textbf{the generative complexity of}} $f$ (with respect to $\mathcal B$).
\bigskip

\ind \textbf{Example}: Let \F\ be an infinite field and consider the field \F\ itself as an \textit{\F-algebra}, denoted $\mathscr A$.
Then the identities of $\mathscr A$ are all the polynomials from \freea\ that evaluate to $0$ under every  assignment from \F\ to the variables \(X\). Namely, these are the (non-commutative) polynomials that are identically zero polynomials \emph{when considered as commutative polynomials}. For instance, $x_1x_2-x_2 x_1$ is a non-zero polynomial from \freea\ which is an identity over $\mathscr A$.

It is not hard to show that the \textit{basis }of the algebra $\mathscr A$ is the \textit{commutator }$x_1 x_2 -x_2 x_1$, denoted $[x_1,x_2]$. In other words, every identity of $\mathscr A$ is generated (in the two-sided ideal) by substitution instances of the commutator. Considering $Q_{\{[x_1,x_2]\}}$, we can now ask what is $Q_{\{[x_1,x_2]\}}(x_1x_3-x_3x_1+x_2x_3-x_3x_2)$? The answer is $1$ since we need only one substitution instance of the commutator: $(x_1+x_2)x_3-x_3(x_1+x_2)=x_1x_3-x_3x_1+
x_2x_3-x_3x_2$.
\bigskip

 Hrube\v s \cite{Hru11} showed the following lower bound (using a slightly different terminology):

\setcounter{tmpt}{\thetheorem}
\begin{theorem}[Hrube\v s \cite{Hru11}]\label{thm:Hrubes} For any field and every $n$, there exists an identity $f\in\freea$ of $\mathscr A$  with $n$ variables, such that $Q_{\{[x_1,x_2]\}}(f) = \Omega(n^2)$.
\end{theorem}

It is also not hard to show that $Q_{\{[x_1,x_2]\}}(f) = O(n^2)$ for any identity $f$.

\para{Lower bound on the complexity of generating matrix identities.}

%

An \textit{algebra with polynomial identities, }or in short a \textbf{\textit{PI-algebra}} (PI stands for  Polynomial Identities), is  simply an \F-algebra that has a non-trivial identity, that is, there is a nonzero $f\in\freea$ that is an identity of the algebra.

Let us treat (the \F-algebra) $\F$ as the matrix algebra Mat$_1(\F)$ of $1\times 1$ matrices with entries from \F.
We shall exploit results about the structure of the identities of matrix algebras and the general theory of PI-algebras to completely generalize Hrube\v s \cite{Hru11} lower bound above (excluding the case $d=2$), from a lower bound of \(\Omega(n^2)\) for generating identities of Mat$_1(\F)$ to a lower bound of $\Omega(n^{2d})$ for generating identities of \matd, for any $d>2$ and any field \F \ of characteristic 0:

\begin{main-lower-bound}[Lower bound on generative complexity]
Let \F\ be any field of characteristic 0.
For every natural number $d>2$ and  every finite basis \(\mathcal B\) of the  identities of  \matd, there exists an identity \(f\) over \matd\ of degree $2d+1$ with $n$ variables, such that $Q_{\mathcal B}(f)=\Omega(n^{2d})$.
\end{main-lower-bound}
%

Notice that similar to \cite{Hru11}, the lower bound in this theorem is \emph{non-explicit}.
We  do not know of an upper bound (in terms of $n$) that holds on  $Q_{\mathcal B}(f)$, for every identity  $f$ with \(n\) variables.

The main lower bound (Theorem \ref{thm:main_lb_on_matrix_proofs})  is  a corollary of the following theorem (proved by simple induction on the number of lines in a \PMd-proof):
\begin{theorem*}\label{prop:Pmatd-connection}
For every identity $F=0$, where $F$
is a non-commutative circuit that computes a non-commutative polynomial $f$ which is an identity of \matd, the number of lines of a \PMd-proof of $F=0$ is lower bounded up to a constant factor (depending on the choice of finite basis \(\mathcal B\)) by $Q_{\mathcal B}(f)$.
\end{theorem*}

\para{Overview of the proof of Theorem \ref{thm:main_lower_bound}.}

The study of algebras with polynomial identities is a fairly developed subject in algebra (see the monographs by Drensky \cite{Dre99} and Rowen \cite{Row80} on this topic). Within it, perhaps the most well studied topic is about  the identities of matrix algebras. In particular, the well-known theorem of Amitsur and Levitzky from 1950 \cite{AL50} is the following:

\begin{A-LThm}[\cite{AL50}]
Let $\mathcal S_d$ be the permutation group on $d$ elements and let $S_d(x_1,x_2,\ldots, x_{d})$ denote the \textbf{standard identity} of degree $d$ as follows:
$$S_{d}(x_1,x_2,\ldots, x_{d}):=
        \sum_{\sigma\in \S_{d}}sgn(\sigma)\prod_{i=1}^{d}x_{\sigma(i)}.$$
Then, for any  natural number $d$ and any field \F\ (in fact, any commutative ring) the standard identity  $S_{2d}(x_1,x_2,\ldots, x_{2d})$ of degree $2d$ is an identity of  \matd.
\end{A-LThm}

Theorem \ref{thm:main_lower_bound} is proved in several steps, but the main argument can be divided into two main parts, described as follows: \smallskip\

\para{Part 1:} Here we use the Amitsur-Levitzki Theorem: we show that when $\mathcal  E =\{S_{2d}(x_1,\ldots,x_{2d})\}$  there exists an $f\in\freea$ with $2n$ variables and degree $2d+1,$ such that $Q_\mathcal E(f)=\Omega(n^{2d})$.
To this end, we generalize the method in  \cite{Hru11} to ``higher dimensional commutativity axioms": using a counting argument we show the existence of $n$ special polynomials (we call \textit{s-polynomials}; see Definition \ref{def:s-poly}) $P_1,P_2,\ldots,P_n$ over  $n$ variables and each of degree $2n$ such that $Q_{S_{2d}}(\nx{P})=\Omega(n^{2d})$ (see Lemma \ref{lem:exist_for_nP}). Then, we combine the $n$ s-polynomials into a single polynomial $\dot{P}$ with degree $2d+1,$ by adding $n$ new variables, such that $Q_{S_{2d}}(\dot{P})=\Omega(Q_{S_{2d}}(\nx{P}))$.

While \cite{Hru11} uses the commutator $[x,y]$ to define the s-polynomials, we consider the higher order commutativity axiom $S_{2d}$ instead. It is possible to show that  $S_{2d}$ has sufficient properties for the lower bound as the commutator $[x,y]$ (see Lemmas \ref{fac:S_2d-equal-zero-when-constant}, \ref{lem:s-poly-linear-property}, \ref{lem:transfer-polynomials}).

\para{Part 2:} Note that $\mathcal E=\{S_{2d}(x_1,\ldots,x_{2d})\}$ is \textit{not }a basis of \matd, namely there are identities of \matd \ that are not generated by substitution instances of $S_{2d}$ (also notice that $Q_{\mathcal B}(f)$ can be  defined for any $\mathcal B\subseteq \freea$). The second part in the proof of Theorem \ref{thm:main_lower_bound} is dedicated to showing that when $d>2$, for \textit{all finite bases $\mathcal B$ of the identities of \matd} the following holds for the hard identity $f$ considered in the theorem: $Q_{\mathcal B}(f)<c\cdot Q_{\mathcal E}(f)$ for some constant $c$.


For this purpose, we find a special set $\mathcal B'\subseteq \freea$ which serves as an  ``intermediate'' set between $\mathcal B$ and $\mathcal E,$ such that $\mathcal B$ is generated by $\mathcal B'$, and all the polynomials in $\mathcal B'$  that contribute to the generation of the hard instance $f$ can be generated already by $\mathcal E$. We then show (Lemma \ref{lem:special-basis}) that for any basis $\mathcal B$, there is a specific set $\mathcal B'$ of polynomials of a special form, namely, \textit{multi-homogenous commutator polynomials }(Definition \ref{def:commutator_identity}), that can generate $\mathcal B$. Based on the properties of multi-homogenous commutator polynomials, we show that, for the hard instance $f$, only the generators of degree at most $2d+1$ in $\mathcal B'$  can contribute to the generation of $f$ (Lemma \ref{lem:2d+2-cant-generate}). We then prove that  when $d>2$, all the generators of degree at most $2d+1$ in $\mathcal B'$ can be generated by $\mathcal E$ (this is where we use the assumption that $d>2$ (see Lemma \ref{lem:generated-by-S_2d})). We thus get the conclusion $Q_{\mathcal B'}(f)<c \cdot Q_{\mathcal E}(f)$, when $d>2$.\bigskip


A very interesting feature of our proof (and theorem), is that it is in fact an \textit{open problem }to describe bases of the identities of \matd, for any $d>2$. For the case $d=2$ the basis is known by a result of Drensky \cite{Dre81} (see  Section \ref{sec:system-diff-algebras}). However,  a highly nontrivial result of Kemer \cite{Kem87}, shows that for any natural $d$ \emph{there exists} a finite basis for \matd. Our proof shows roughly that for the hard instances $f$ in Theorem \ref{thm:main_lower_bound} no generators different from the  $S_{2d}$ generators can contribute to the generation of $f$.
\smallskip

%
%
%
%


%
\section{Towards strong lower bounds on (full) arithmetic proofs}
\label{sec:intro_hard_identities}

Here we continue the study of matrix identities as hard proof complexity instances, and set out a program to establish  lower bounds on arithmetic proofs. We present two conjectures, interesting by themselves: one
about non-commutative arithmetic circuit complexity and the
other about proof-complexity, based on which up to exponential-size lower bounds on arithmetic proofs (in terms of the non-commutative circuit-size of the
identity proved) follow. We discuss in details these
conjectures and the parameters the are needed for different
kinds of lower-bounds.

Informally, the two conjectures are as follows (recall the complexity measure $Q_{\mathcal B}(f)$ from Sec.~\ref{sec:ovrv:proof_of_main_lower_bound}, counting the minimal number of substitution instances of generators from a basis $\mathcal B$ needed to generate an identity $f$):

\newtheorem*{conj_circ_size}{Conjecture I}

\begin{conj_circ_size} (Informal)
There exist non-commutative arithmetic circuits of small size that compute matrix identities of high generative complexity.
\end{conj_circ_size}

\newtheorem*{main_conj}{Conjecture II}
\begin{main_conj} (Informal)
Proving matrix identities by reasoning with polynomials whose variables $X_1,\ldots, X_n$ range over matrices is \emph{as
efficient} as proving matrix identities using polynomials whose variables  range over the \emph{entries} of the matrices $X_1,\ldots, X_n$?
\end{main_conj}


%


\subsection{Towards lower bounds on \PMd\ in terms of arithmetic-circuit size}\label{sec:ovrv:fragment-lower-bounds}

Recall that a non-commutative arithmetic circuit is an arithmetic circuit that has an order on the children of product gates and the product is performed according to this order (see Sec.~\ref{sec:app_arithmetic_circuit}).
To get a size lower bound on \PMd\ proofs in terms of the circuit equations proved, we need to assume the existence of non-commutative arithmetic circuits of small size that compute matrix identities of high generative complexity:

\bigskip

\newtheorem*{C1}{Conjecture I}
\begin{boxedminipage}[c]{1.\textwidth}
\vspace{5pt}
\begin{C1}
For some fixed $d\ge 1$, there exists a family of identities  $f_n\in\freea$ of \matd, with $n$ variables, such that $Q_{\mathcal B}(f_n)=\Omega(n^d)$, for some basis $\mathcal B $ of the identities of \matd,
\textit{and} such that
$f_n$ has a non-commutative arithmetic circuit of size  $O(n^r)$, for some constant $r<d$.
\end{C1}\vspace{0pt}
\end{boxedminipage}
\bigskip

Assuming the veracity of the above conjecture we obtain the following lower bound:\medskip

\ind
\textbf{Polynomial lower bounds on \PMd-proofs (assuming Conjecture I):} \textit{There exists a family of identities $f_n$ of \matd\ whose non-commutative arithmetic circuit-size is $s_n$ but every \PMd-proof of $f_n$ has size  $\Omega(s_n^{d-r})$.}
\bigskip

Note that we know by Theorem \ref{thm:main_lower_bound} that the lower bound in Conjecture I is true for any $d>2$ and for some specific family  $f_n$. But we do not know whether this specific $f_n$ has small circuits, as required in Conjecture I. 

\bigskip

\subsection{Towards polynomial-size  lower bounds on full arithmetic  proofs}

Here we consider the possibility that the  arbitrary polynomial-size lower bounds on matrix identities proofs \PMd\ transfer  to arithmetic proofs \PC(\F) lower bounds.





The natural way to formalize Conjecture II mentioned  informally above  is via the following  translation: consider a nonzero identity $f$ of \matd, for some $d>1$. Then $f$ is a nonzero non-commutative polynomial in \freea. If we substitute each (matrix) variable $x_i$ in $f$ by a \dbyd\ matrix of \emph{entry-variables} $\{x_{ijk}\}_{j,k\in[d]}$, then  $f$ corresponds to $d^2$ \textit{commutative }zero polynomials: $f=0$ says that for  every $(i,j)$ and for every possible assignment of field \F\ elements to the $(i,j)$-entry of each of the matrix variables in $f$ (when the product and addition of matrices are done in the standard way) the $(i,j)$-entry evaluates to $0$.
Accordingly, let $F$ be a non-commutative circuit computing $f$. Then under the above substitution of $d^2$ entry-variables to each variable in $F$, we get $d^2$  non-commutative circuits, each computing the zero polynomial \emph{when considered as commutative polynomials} (see Definition \ref{def:double-bracket}). We denote the set of $d^2$ circuits   corresponding to the identity $F$ by $\llbracket F\rrbracket _d  $ (and we  extend it naturally to equations between circuits: $\llbracket F=G\rrbracket _d$).

\para{Example:} Let $d=2$ and let $f=xy-yx$ (it is obviously not an identity of \mattwo, but we use it only for the sake of example). And let $F=xy-yx$ be the corresponding circuit (in fact, formula) computing $f$. Then we substitute matrices for $x,y$ to get:
$$\begin{pmatrix}x_{11} & x_{12} \\
x_{21} & x_{22} \\
\end{pmatrix} \cd
\begin{pmatrix}y_{11} & y_{12} \\
y_{21} & y_{22} \\
\end{pmatrix}
-
\begin{pmatrix}y_{11} & y_{12} \\
y_{21} & y_{22} \\
\end{pmatrix}
\cd
\begin{pmatrix}x_{11} & x_{12} \\
x_{21} & x_{22} \\
\end{pmatrix}.$$
And the $(1,1)$-entry non-commutative circuit (in fact formula) in  \(\convert{F}\), is:
$$(x_{11}y_{11}+x_{12}y_{21})-(y_{11}x_{11}+y_{12}x_{21}).$$

It is not hard to show that $\big|\convert{F}\big|=O\left(d^3 |F| \right)$, for every non-commutative circuit $F$ (where
$\big|\convert{F}\big|$ is the total sizes of all circuits in $\convert{F}$ and $|F|$ is the size of $F$).
We denote by
$$\big|\vdash_{\PC(\F)}\convert{F=0}\big|$$ the minimal size of a \PC(\F) proof that contains (as proof-lines) all the
circuit-equations in $\convert{F=0}$.

\bigskip

\begin{boxedminipage}[c]{1\textwidth}
\begin{main-open}  Let $d$ be a positive natural number and let $ \mathcal B$ be a (finite) basis of the  identities of \matd. Assume that $f\in\freea$ is an identity of \matd, and let $F$ be a non-commutative arithmetic circuit computing $f$. Then, the minimal number of lines in a $\PC(\F)$ proof of the collection of  $d^2$ (entry-wise) equations $\llbracket F=0 \rrbracket_d$  corresponding to $F$, is lower bounded (up to a constant factor) by $Q_{\mathcal  B}(f)$. And in symbols:
\begin{equation}\label{eq:main-open-prob}
\big|\vdash_{\PC(\F)}\convert{F=0}\big| = \Omega(Q_{\mathcal B}(f)).
\end{equation}
\end{main-open}
\end{boxedminipage}
\bigskip

\medskip

The conditional lower bound we get is now similar to that in Section \ref{sec:ovrv:fragment-lower-bounds}, except that it holds for \PC(\F) and not only for \textit{fragments} of \PC(\F):
\bigskip

\ind\textbf{Polynomial lower bounds on arithmetic proofs \PC(\F) (assuming Conjectures I and II):} \textit{There exists a family of identities $f_n$ of \matd\ whose non-commutative arithmetic circuit-size is $s_n$ but every $\PC(\F)$-proof of $f_n$ has size  $\Omega(s_n^{d-r})$.}
\bigskip

We also present a   \emph{propositional version} of Conjecture II, by  considering $\F$ to be $GF(2)$, adding to \PC(\F) the Boolean axioms $x_i^2+x_i=0$ and considering matrix identities for \matd\ (see Section \ref{sec:prop-vers}).

\subsection{Towards \textit {exponential-size} lower bounds on arithmetic proofs}

Assuming Conjecture II above holds (i.e., Equation \ref{eq:main-open-prob}), we show under which further conditions one gets \emph{exponential-size} lower bounds on arithmetic proofs \PC(\F).
The idea is  to take \textit{the dimension $d$ of the matrix algebras as a parameter by itself.}
For this we need to set up the assumptions more carefully:\bigskip

\ind\textbf{Assumptions:}

\begin{enumerate}

\item \textbf{Refinement of Conjecture II}: Assume that for any $d$ and any basis $\mathcal B_d$ of the identities of \matd\ the number of lines in any $\PC(\F)$ proof of $\llbracket F= 0\rrbracket_d$ is at least  $\mathcal C_{\mathcal B_d}\cd Q_{\mathcal B_d}(f)$, where $\mathcal C_{\mathcal B_d}$ is a number depending on \(\mathcal B_{d}\)  and $F$ is a non-commutative arithmetic circuit computing $f$ (this is the same as Conjecture II except that now $\mathcal C_{\mathcal B_d}$ is not a constant).

\item Assume that for any sufficiently large $d$ and any basis $\mathcal B_d$ of the identities of \matd, there exists a number $c_{\mathcal B_d}$ such that for all sufficiently large \(n\) there exists an identity $f_{n,d}$ with  $Q_{\mathcal B_d}(f_{n,d})\ge c_{\mathcal B_d}\cd n^{2d}$. (The existence of such identities are known from our unconditional lower bound.)

\item Assume that for the \(c_{\mathcal B_{d}}\) in item 2 above:   $c_{\mathcal B_d}\cd \mathcal C_{\mathcal B_d}= \Omega\left(\frac{1}{\poly(d)}\right)$.

\item \textbf{(Variant of) Conjecture I}: Assume that the non-commutative arithmetic circuit size of $f_{n,d}$ is  at most $\poly(n,d)$.
\end{enumerate}

\ind\textbf{Corollary (assuming Assumptions 1-4 above)}:
There exists a polynomial size (in $n$) family of identities between non-commutative arithmetic circuits, for which any \PC(\F) proof requires exponential $2^{\Omega(n)}$ number of proof-lines.

\begin{proof} By the assumptions, every $\PC(\F)$-proof of $\llbracket f_{n,d}=0 \rrbracket_d$ has size at least $c_{\mathcal B_d}\cd
\mathcal C_{\mathcal B_d}\cd n^{2d}$.  Consider the family $\{ f_{n,d}\}_{n=1}^\infty$, \textit{where $d$ is a function of $n$}, and we take  $d=n/4$. Then, we get the following lower bound on the number of lines in any    $\PC(\F)$-proof of the family $\{ f_{n,d}\}_{n=1}^\infty$:
\[c_{\mathcal B_d}\cd \mathcal C_{\mathcal B_d}\cd n^{2d}=\frac{1}{\poly(n/4)}n^{n/2}= 2^{\Omega(n)}
,
 \]
which (by Assumption 4) is \textit{exponential} in the arithmetic circuit-size of the identities \(f_{n,d}\) proved.
\end{proof}
\para{Justification of assumptions.}

We wish to justify to a certain extent the new Assumptions 3 above (which lets us obtain the exponential lower bound). We shall use the special hard polynomials $f$ that we proved exist in Theorem \ref{thm:main_lower_bound} for this purpose.

First, note that Assumption 2 holds for these $f$'s, by Theorem \ref{thm:main_lower_bound}.
  In Section \ref{sec:exponential-lower-bounds} we show that the function $c_{\mathcal B_d}$ for these $f$'s does not decrease too fast.
And we use this fact to get the following (conditional exponential lower bound):

\Comment{
That is, for the polynomial $f$ we can prove the following:
 $$Q_{\mathcal B_{n/4}}(f)=\Omega\left(\frac{2^n}{n^{5/2}\ln n}\right).$$
}


\begin{proposition*}
Suppose \emph{Assumption 1} above holds (refinement of Conjecture II) and assume that $\mathcal C_{\mathcal B_{n/4}}=
\Omega(1/{\rm poly}(n))$. Then, there exists a  family of
non-commutative circuits $\{F_n\}_{n=1}^\infty$ (computing the family of polynomials $\{f_{n,\frac{n}{4}}\}_{n=1}^\infty$)
such that  the number of lines in any $\PC(\F)$-proof of $\llbracket F_n= 0\rrbracket_{n/4}$ is at least $2^{\Omega(n)}$.
%
%
\end{proposition*}


Note that this will give us an exponential-size lower bound on \PC(\F) proofs only if moreover  the arithmetic circuit size of $\{F_n\}_{n=1}^\infty$ is small enough (e.g., if Assumption 4 above holds).

\section{Concluding remarks}
This work originates from the fundamental  goal of establishing lower bounds on strong proof systems. Our focus was on arithmetic proofs which serve as a useful \cite{HT12} analogue of propositional Extended Frege proofs. Along the way, we have discovered an interesting hierarchy within arithmetic proofs: a hierarchy of sound and complete proof systems for matrix identities of increasing dimensions. In this hierarchy we have been able to establish unconditional nontrivial size-lower bounds (in terms of the number of variables in the identities proved).

We then used these results, together with two seemingly natural conjectures about non-commutative arithmetic circuits and proof complexity, to propose matrix identities as hard candidates for strong proof systems. We showed that using these two conjectures, one can obtain up to exponential-size lower bounds (in terms of the circuit-size of the identities proved).

Proving lower bounds on strong (propositional) proof systems is a fundamental open problem in the theory of computing; nevertheless, it is in fact not clear whether such lower bounds are beyond current techniques (in contrast to other fundamental hardness problems in complexity, such as explicit Boolean circuits lower bounds). In light of this, and the fact that almost no hard candidates for strong proof systems are currently known  (see \cite{BBP95,Kra10-forcing}), it seems  that an important \textit{conceptual}, so to speak, contribution of this paper, is to supply such new hard candidates in the form of matrix identities. Moreover, as our work partially demonstrates, such matrix identities  have structure that is helpful in proving proof
complexity lower bounds.

%
%
%
%
%

%
\section{Relation to previous work}
\para{Relation to previous work by Hrube\v s \cite{Hru11}.}
The problem of proving \textit{quadratic} size lower bounds on arithmetic proofs $\PC$ was considered by Hrube\v{s} in \cite{Hru11}. The work in \cite{Hru11} gave several conditions and open problems, under which, quadratic size lower bounds on arithmetic proofs would follow (and further, showed that the general framework suggested may have potential, at least in theory, to yield Extended Frege quadratic-size lower bounds). The current work can be viewed as an attempt to extend the approach suggested in Hrube\v s \cite{Hru11}, from an approach suitable for proving up to  $\Omega(n^2)$ size lower bounds on $\PC$ proofs, (and potentially Extended Frege proofs) to an approach for proving much stronger lower bounds,  namely an $\Omega(n^d)$ lower bound on $\PC(\F)$ proofs, for every positive $d>2$ and for every zero characteristic field $\F$; and under stronger assumptions, exponential  $2^{\Omega(n)}$  lower bounds on $\PC(\F)$ proofs (and  similarly, potentially on Extended Frege proofs).

\para{Relation to other previous works.}
Apart from the connection to  \cite{Hru11}, we may consider the relation of the current work to the work of  Hrube\v{s} and Tzameret \cite{HT12} that obtained polynomial-size  (arithmetic and propositional) proofs for certain identities concerning matrices. As far as we see, there are no direct relations between these two works: in the current work we are studying matrix identities whose number of matrices (i.e., variables) grows with the number of variables $n$ (if the number of matrices in the matrix identities over \matd\ is $m$ then the number of variables in the translation of the identities to a set of $d^2$ identities is $d^2\cd n$).
 Whereas in \cite{HT12} the number of matrices was fixed and only the dimension of the matrices grows.

Note also that the matrix identities studied in \cite{HT12} are not even translations (via $\llbracket \cd \rrbracket $) of matrix identities over \matd. For instance consider the identity $\det(A)\cd\det(B)=\det(AB)$ from \cite{HT12}, where $A$ and $B$ are $2\times 2$ matrices. Then we get that:
$$ \,
\det\begin{pmatrix}
a & b \\
c & d \\
\end{pmatrix}
\cd
\det\begin{pmatrix}
e & f \\
g & h \\
\end{pmatrix} =
\det\begin{pmatrix}
ae+bg & af+bh \\
ce+dg & cf+dh\\
\end{pmatrix}$$
is equal
to
$
(ad-bc)\cd(eh-fg)=(ae+bg)(cf+dh)-(af+bh)(ce+dg).
$
But notice that, e.g., in our translation of a  matrix identity over \matd, two variables that correspond to the same matrix cannot multiply each other, while in the example above, $a$ multiplies $c$ and $b$ multiplies $d$, though they are entries of the same matrix.

\newpage

\appendix
\section*{\huge Technical appendix}

\section{Formal preliminaries}
\subsection{Algebras with polynomial identities}
For a natural number $n$, put $[n]:=\set{1,2,...,n}$. We use lower case letters $a,b,c$ for constants from the underlying field, $x,y,z$ for variables and $\overline x,\overline y,\overline z$ for vectors of variables, $f,g,h,\ell$ or upper case letters such as $A,B,P,Q$ for polynomials and $\overline f,\overline g,\overline h,\overline \ell, \overline A,\overline B,\overline P, \overline Q$, for vectors of polynomials (when the arity of the vector is clear from the context).

A polynomial is a formal sum of monomials, where a monomial  is a product of (possibly non-commuting) variables and a constant from the underlying field.
%
%
For two polynomials $f(x_1,\ldots,x_n)$ and $g$ we say that  $g$ \emph{is a substitution instance of $f$} if $g=f(h_1,\ldots,h_n)$
for some polynomials $h_1,\ldots,h_n$; and we sometimes denote $f(h_1,\ldots,h_n)$ by $f(\overline h)$.
For a polynomial $f(\nx{x})\in\freea$, $f\big |_{x_{i_1}\leftarrow g_{i_1},\ldots, x_{i_k}\leftarrow g_{i_k}}$ denotes the polynomial  that replaces $x_{i_1},\ldots, x_{i_k}$ by  $g_{i_1},\ldots, g_{i_k}$ in $f,$ respectively, where $g_{i_1},\ldots, g_{i_k}\in \freea, i_1 ,\ldots, i_k$ are distinct numbers from $[n]$ and $k\in[n]$.

For a vector $\t H$ of polynomials $H_1 ,\ldots, H_k\in\freea$ where $k$ is positive integer, we also use the notation  $\overline {H}|_{ H_j\leftarrow f}$, to denote the vector of polynomials that replace the $j^{th}$ coordinate $H_j$ in $\t {H}$ by a polynomial $f\in \freea$, where $j\in[k]$.
\iddocomment{but we say that lower case letters like f denotes constants?}

\begin{definition}

Let A be a vector space over a field $\F$ and $\,\cd: A\times A \rightarrow A$ be a distributive multiplication operation. If $\cd$ is associative, that is, $a_1 \cd (a_2 \cd a_3) = (a_1\cd a_2)\cd a_3$ for all $a_1,a_2,a_3$ in $A$, then the pair $(A,\cd)$ is called an \textbf{associative algebra over $\F$}, or an $\F$\textbf{-algebra}, for short.\footnote{In general an \F-algebra can be non-associative, but since we only talk about associative algebras in this paper we use the notion of $\F$-algebra to imply that the algebra is associative.}
\end{definition}
%

Perhaps the most prominent example of an $\F$-algebra is the algebra of $d\times d$ matrices, for some  positive natural number $d$, with entries from $\F$ (with the usual addition and multiplication of matrices). We denote this algebra by \matd. Note indeed that \matd\ is an associative algebra but not a commutative one (i.e., the multiplication of matrices is non-commutative because $AB$ does not necessarily equal $BA$, for two $d\times d$ matrices $A,B$).

\begin{definition}
Let $\freea$ denote the associative algebra of all polynomials such that the variables $X:=\{x_1,x_2,\ldots \}$ are non-commutative with respect to multiplication. \iddocomment{Should be finite? Countable?} We call $\freea$ the \textbf{free algebra (over $X$)}.
\end{definition}
For example, $x_1 x_2-x_2 x_1+x_3 x_2 x_3^2-x_2 x_3^3,\  \, x_1 x_2-x_2 x_1$ and $0$ are three distinct polynomials in $\freea$.

Note that the set $\freea$ forms a non-commutative ring. We sometimes call $\freea$ \emph{the ring of non-commutative polynomials} and call the polynomials from $\freea$ \emph{non-commutative polynomials}. Throughout this paper, unless otherwise stated, a polynomial is meant to be a non-commutative polynomial, namely a polynomial from the free algebra  $\freea$.

We now introduce the concept of a \textit{polynomial identity algebra}, PI-algebra for short:

\begin{definition}\label{def:PI-algebra}
Let $A$ be an \F-algebra. An \textbf{identity of $A$} is a polynomial $f(x_1,...,x_n)\in\freea$ such that: \[
f(a_1,...,a_n)=0,    \mbox{ for all } a_1,...,a_n\in A.
\]
A \textbf{PI-algebra} is simply an algebra that has a non-trivial identity, that is, there is a nonzero $f\in\freea$ that is an identity of the algebra.
\end{definition}

For example, every \textit{commutative} \F-algebra $A$ is also a PI-algebra: for any $a,b\in A$, it holds that $ab-ba = 0$, and so $x_i x_j - x_j x_i $ is a nonzero polynomial identity of $A$, for any positive $i\neq j\in \N$. A  concrete example of a commutative algebra is the usual ring of (\emph{commutative}) polynomials with coefficients from a field \F\ and variables $X=\{x_1,x_2,\ldots\}$, denoted usually $\F[X]$.

An example of an algebra that is \textit{not} a PI-algebra is the free algebra \freea\ itself. This is because a nonzero polynomial $f\in\freea $ cannot be an identity of \freea\ (since the assignment that maps each variable to itself does not nullify $f$).

A \emph{two-sided ideal} $I$ of an \F-algebra $A$ is a subset of $A$ such that for any (not necessarily distinct) elements $f_1,...,f_n$ from $I$ we have $\sum_{i=1}^n g_i\cd f_i \cd h_i \in I$, for all $g_1,...,g_n, h_1,...,h_n\in A$.

\begin{definition}\label{def:T-ideal} A \textbf{T-ideal} $\mathcal T$ is a two-sided ideal of \freea\ that is closed under all endomorphisms\footnote{An algebra endomorphism of $A$ is an (algebra)
homomorphism $A\to A$.}, namely, is closed under all substitutions of variables by polynomials.
\end{definition}
In other words, a T-ideal is a two-sided ideal $\mathcal T$, such that if $f(x_1,...,x_n)\in \mathcal T$ then $f(g_1,...,g_n)\in \mathcal T$, for any $g_1,...,g_n\in \freea$.

It is easy to see the following:
\begin{fact} \label{fac:identities-are-T-ideal}
The set of identities of an (associative) algebra is a T-ideal.
\end{fact}

A basis of a T-ideal $\mathcal T$ is a set of polynomials whose substitution instances generate $\mathcal T$ \textit{as an ideal}: 
\begin{definition}
Let $B\subseteq \freea$ be a set of polynomials and let $\mathcal T$ be a T-ideal in \freea. We say that $B$ is \textbf{a basis for $\mathcal T$} or that \textbf{$\mathcal  T$ is generated as a T-ideal by $B$}, if every $f\in \mathcal T$ can be written as:
\[
f=\sum_{i\in I} h_i\cd B_i(g_{i1},...,g_{in_i})\cd \ell_i\,,
\]
for $h_i,\ell_i,g_{i1},...,g_{in_i}\in\freea$ and $B_i\in B$ (for all $i\in I$).
\end{definition}
Given $B\subseteq \freea$, we write $T(B)$ to denote the T-ideal generated by $B$. Thus, a T-ideal $\mathcal T$ is generated by $B \subseteq \freea$ if $\mathcal T = T(B)$.

\begin{examples}
$T(x_{1})$ is simply the set of all polynomials from $\freea$. $T(x_1 x_2-x_2 x_1)$ is the set of all non-commutative polynomials that are zero if considered as commutative polynomials.
\end{examples}

Note that the concept of a T-ideal is already somewhat reminiscent of logical proof systems, where generators of the T-ideal $\mathcal T$ are like axioms schemes and generators of a two-sided ideal containing $f$ are like  substitution instances of the axioms.

\iddocomment[inline]{explain the last example  more precisely/formally.}

\iddocomment[inline]{Define homogenous polynomials and write this as a notation in the definition}
\iddocomment{no numbers in notations ?}

A polynomial is \textit{homogenous} if all its monomials have the same total degree. Given a polynomial $f$, the \emph{homogenous part of degree $j$} of $f$, denoted $f^{(j)}$ is the sum of all monomials with total degree $j$.  We write $\degr{C}{j}$ to denote the $j$th-homogeneous part of the circuit $C$ and the vector $\degr{\overline{C}}{j}$  denotes the vector consisting of the $j$th-homogeneous parts of the circuits $C_1,C_2,\ldots,C_{2d}$.

\iddocomment[inline]{Do you really need homogenous part of circuits or is it enough to have homogenous part of a polynomial?
"$\degr{C}{j}$ denotes the $j$th-homogeneous part of the circuit $C$ and the vector $\degr{\overline{C}}{j}$  denotes the vector consisting of the $j$th-homogeneous parts of the circuits $C_1,C_2,\ldots,C_{2d}$."
}

\smallskip

\begin{definition}
$S_d(x_1,x_2,\ldots, x_{d})$ denotes the  \textbf{standard identity} of degree $d$ as follows:
$$S_{d}(x_1,x_2,\ldots, x_{d}):=\sum_{\sigma\in \S_{d}}sgn(\sigma)\prod_{i=1}^{d}x_{\sigma(i)}\,,$$
where $\S_d$ denotes the symmetric group on $d$ elements and $sgn(\sigma)$ is the sign of the permutation $\sigma$. \iddocomment{sym. group of degree or of ``order'' d ?}
\end{definition}

For $n$ polynomials $\nx{f}$ where $n\geq 2,n\in \Z$, we define the \emph{\textbf{generalized-commutator}} $[\nx{f}]$ as follows:
$$[f_1,f_2]:=f_1f_2-f_2f_1, ~~~~\hbox{(in case $n=2$)}$$
$$\hbox{and}~~~~[f_1,\ldots,f_{n-1},f_n]:=[[f_1,\ldots,f_{n-1}],f_n],\;  ~~~\hbox{for $n>2$}.$$

A polynomial $f\in \freea$ with $n$ variables is homogenous  \emph{with degrees $(1,\ldots,1)$}  ($n$ times) if in every monomial the power of every variable $x_1,\ldots,x_n$ is precisely 1. In other words, every monomial is of the form $\alpha\cd \prod _{i=1}^n x_{\sigma(i)}$, for some permutation $\sigma$ of order $n$ and some scalar $\alpha$. For the sake of simplicity, we shall talk in the sequel about \textit{\textbf{ polynomial of degree $n$}}, when referring to  polynomial with degrees $(1,\ldots,1)$ ($n$ times). Thus, any  polynomial with $n$ variables is homogenous of total-degree  $n$.

\subsection{Arithmetic circuits}\label{sec:app_arithmetic_circuit}

\begin{definition}
Let $\F$ be a field, and let $X=\set{\nx{x}}$ be a set of input variables. An \textbf{arithmetic (or algebraic) circuit} is a directed acyclic graph, where the in-degree of nodes is at most $2$. Every leaf of the graph (namely, a node of in-degree 0) is labelled with either an input variable or a field element. Every other node of the graph is labelled with either $+$ or $\times$(in the first case the node is a sum-gate and in the second case a product-gate). Every edge in the graph is labelled with an arbitrary field element. A node of out-degree $0$ is called an output-gate of the circuit.
\end{definition}

Every node and every edge in an \emph{arithmetic circuit} computes a polynomial in the commutative polynomial-ring $\F[X]$ in the following way. A leaf just computes the input variable or field element that labels it. 
the sum of the polynomials computed by the two edges that reach it. A product-gate computes the product
of the polynomials computed by the two edges that reach it. We say that a polynomial $g \in \F[X]$ is computed by the circuit if it is computed by one of the circuit's output-gates.

The size of a circuit $\Phi$ is defined to be the number of edges in $\Phi$, and is denoted by $|\Phi|$.

\begin{definition}
Let $\F$ be a field, and let $X=\set{\nx{x}}$ be a set of input variables. A \textbf{non-commutative arithmetic circuits} is
similarly to the arithmetic circuits defined above, with the following additional feature: given any $\times$-gate of fanin $2$, its children are labeled by a fixed order.
\end{definition}
Every node and every edge in a \emph{non-commutative arithmetic circuit} computes a noncommutative polynomial in the free algebra $\freea$ in exactly the same way as the arithmetic circuit does, except that at each $\times-gate$, the ordering among the children is taken into account in defining the polynomial computed at the gate.

The size of a noncommutative circuit $\Phi$ is also defined to be the number of vertices in $\Phi$, and is denoted by $|\Phi|$.

\section{The complexity measure}
\iddocomment{might not be appropriate name}
Let $A$ be a PI-algebra (Definition \ref{def:PI-algebra}) and let  $\mathcal T$ be the T-ideal (Definition \ref{def:T-ideal}) consisting of all identities of $A$ (see Fact \ref{fac:identities-are-T-ideal}). Assume that $B$ is a basis for the  T-ideal $\mathcal T$, that is, $T(B)=\mathcal T$.  Then every $f\in\mathcal T$ is a consequence \iddotodo{define consequence} of $B$, namely, can be written as a linear combination of substitution instance of polynomials from $B$ as follows:
\begin{equation}\label{eq:def-complexity}
f=\sum_{i\in I} h_i\cd B_i(g_{i1},...,g_{in_i})\cd \ell_i\,,
\end{equation}
for $h_i,\ell_i,g_{i1},...,g_{in_i}\in\freea$ and $B_i\in B$ (for all $i\in I$).

A very natural question, from the complexity point of view, is the following: \textit{What is the minimal number of distinct substitution instances $B_i(g_{i1},\ldots,g_{in_i})$ of generators from $B$ that must occur in (\ref{eq:def-complexity})?} Or in other words, \textit{how many distinct substitution instances of generators are needed to generate $f$ above?}
\iddocomment{Note it's different than the minimal $|I|$ such that (\ref{eq:def-complexity}) holds,
because the same substitution instance can occur (necessarily) more than once in the sum (because non-commutative multiplication).}

Formally, we have the following:
\begin{definition}[$Q_{B}(f)$]\label{def:Q-measure}
For a set of polynomials  $B \subseteq \freea $, define $Q_B(f)$ as the smallest (finite) $k$ such that there exist substitution instances $g_1,g_2,\ldots, g_k$ of polynomials from $B$ with
$$
f\in \langle g_1,g_2,\ldots, g_k \rangle ,
$$
where $\langle g_1,g_2,\ldots, g_k\rangle$ is the two-sided ideal generated by $g_1,g_2,\ldots,g_k$.
\end{definition}

If the set $B$ is a singleton $B=\{h\}$, we shall sometimes
write $Q_{h}(\cd)$ instead of $Q_{\{h\}}(\cd)$.

Accordingly, we extend Definition \ref{def:Q-measure} to  a \textit{sequence }of polynomials and let $Q_{B}(\nx{f})$ be the smallest $k$ such that there exist some substitution instances $g_1,g_2,\ldots, g_k$ of polynomials from $B$ with
$$f_i\in \langle g_1,g_2,\ldots, g_k\rangle, ~~~\hbox{for all } i\in[k].$$

Note that $Q_B(f)$ is interesting only if $f$ is not already in the generating set. Hence, we need to make sure that the generating set does not contain $f$ and the easiest way to do this  (when considering asymptotic growth of  measure) is by stipulating the the generating set is finite. Given an algebra, the question whether there exists a finite generating set of the T-ideal of the identities of the algebra is a  highly non-trivial  problem, that goes by the name \textit{The Specht Problem{}}. Fortunately, for matrix algebras we can use the solution of the Specht problem given by Kemer \cite{Kem87}. Kemer showed that for every matrix algebra $A$ there exists a finite basis of the T-ideal of the identities of $A$. The problem to actually find such a finite basis for most matrix algebras (namely for all values of $d$, for \matd) is open.

We have the following simple proposition (which is analogous to a certain extent to the fact that every two Frege proof systems polynomially simulate each other; see e.g.~\cite{Kra95}):

\iddocomment[inline]{Unclear definition; also it should be a definition and not a notation.}

\begin{proposition}\label{prop:generate-means-less-Q}
Let $A$ be some \F-algebra and let $B_{0}$ and $B_1$ be two \emph{finite} bases for the identities of $A$. Then, there exists a constant $c$ (that depends only on $B_0,B_1$) such that for any identity $f$ of \(A\):
$$
Q_{B_0}(f)\le c \cd Q_{B_1}(f).
$$
\end{proposition}
\begin{proof}
Assume that $B_0=\{A_1,A_2,\ldots,A_k\}$ and $B_1=\{B_1,B_2,\ldots, B_\ell\}$. And suppose that $Q_{B_1}(f)=q$ and $f\in\ideal{B_{i_1}(\overline {g_1}),\ldots,B_{i_q}
(\overline {g_q})}$, for $i_j \in [\ell]$ and where $\overline {g_j}\in\freea$ are the substitutions of polynomials for the variables of $B_{i_j}$. By assumption that both $B_0$ and $B_1$ are bases for $A,$ there exists a constant $r$ such that $B_{i_j} \in \ideal{A_{j_1}(\overline {h_{j_1}}),...,A_{j_r}(\overline {h_{j_r}})}$, for all $j\in[q]$, and where $\overline {h_{j_l}}\in\freea$ are the substitutions of polynomials for the variables of $A_{j_l}$, for any $l\in[r]$ (formally, $r=\max\{Q_{B_0}(B_i)\;:\; i\in[\ell]\}$).

Note that if $B_{i_j}\in\ideal{A_{j_1}(\overline {h_{j_1}}),\ldots,A_{j_r}
(\overline {h_{j_r}})}$, then for any substitution $\overline g_j$ (of polynomials to the variables $X$) we
have $B_{i_j}(\overline{g_j})\in\ideal{\left(A_{j_1}(\overline {h_{j_1}})\right)(\overline {g_j}),\ldots,\left(A_{j_r}
(\overline {h_{j_r}})\right)(\overline {g_j})}$. Thus, every $  B_{i_j}(\overline{g_j})$ is generated by $r$ substitution instances of polynomials from $B_0$, for any $j\in[q]$. Therefore,  $f$ can be generated with at most $r\cd q $ substitution instances of generators from $B_0$, that is,
\begin{equation}\label{eq:propostion_generator_set}
  Q_{B_0}(f)\le r\cd Q_{B_1}(f)~~~~~~\text{where $r=\max\{Q_{B_0}(B_i)\;:\; i\in[\ell]\}$}.
\end{equation}
\end{proof}




\iddocomment[inline]{I rewrote a bit the proof; hopefully it is a bit more clear now.}

\section{Matrix algebras}

\para{Hrube\v s' work.}
For an identity $f$ in a commutative algebra, we define the notation $Q_{\{[x,y]\}}(f)$  as the minimal number of substitution instances of the commutativity axioms $[x,y]=0$ we need to generate $f$ in the two-sided ideal.

For example, $Q_{[x,y]}(x_1x_2-x_2x_1)$ is $1$.
And $Q_{[x,y]}(x_1x_2-x_2x_1+x_1x_3-x_3x_1)$ is also $1$ since the formula $ x_1x_2-x_2x_1+x_1x_3-x_3x_1$ equals $[x_2+x_3,x_1]$.
In \cite{Hru11} it was concluded that there is an identity $f$ with \(n\) variables, such that:
$$Q_{[x,y]}(f)=\Omega(n^2).$$
\bigskip

We wish  to extend this result to matrix algebras. Let \matd\ denote the $d\times d$ matrix algebra over $ \mathbb{F}$, that is,  the set of all $n  \times n$ matrices with entries from $\F$, with the usual operations of matrices.
First of all, we extend the notation $Q_{[x,y]}(f)$, which only count the instances of one axiom\iddocomment{undefined}, to the notation $Q_{A_1,A_2,\ldots,A_n}$ which count the instances of $n$ axioms $A_1=0,A_2=0,\ldots,A_n=0.$


Concerning matrix algebras, the following is the famous Amitsur-Levitzky Theorem:
\begin{A-LThm}[\cite{AL50}]
For any  natural number $d$ and any field \F\ (in fact, any commutative ring) the standard identity  $S_{2d}(x_1,x_2,\ldots, x_{2d})$ of degree $2d$ is an identity of  \matd.
\end{A-LThm}

%
%
%

\iddocomment{Before a reference put a space}
Further, it can be shown that \matd\ does not have identities of degree smaller than $2d$. And that the identities of \matd\ can be \textit{finitely }generated \iddocomment{undefined}\cite{Kem87}. That is, there must be a finite generating set for \matd.
By Proposition  \ref{prop:generate-means-less-Q} no matter which finite generating set $\{A_1,A_2,...,A_k\}$ for \matd\ we choose, the value $Q_{A_1,A_2,...A_k}$ is the same up to a constant factor.


Our main theorem is the following:
\begin{theorem}\label{thm:main_lower_bound}
Let \F\ be any field of characteristic 0.
For every natural number $d>2$ and for every finite basis \(\mathcal B\) of the T-ideal of identities of  \matd, there exists an identity \(P\) over \matd\ of degree $2d+1$ with $n$ variables, such that $Q_{\mathcal B}(P)=\Omega({n \choose 2d})=\Omega(n^{2d})$.
\end{theorem}

It is interesting to point out that although we do not necessarily know what is the (finite) generating set of \matd\ we still can lower bound the number of generators needed to generate certain identities. \iddocomment {This should be stated in intro.}
\iddocomment[inline] {You forgot to say that the field should always be of characteristic zero in this work !}

\subsection{The lower bound}\label{sec:the-lower-bound}

We start by proving a lower bound on  $Q_{S_{2d}}$, that  is, we prove a lower bound on the number of substitution instances of $S_{2d}$ identities needed to generate a certain identity (though $S_{2d}$ is \textit{not} known to be the basis of the T-ideal of the identities over \matd) .
\begin{lemma}
For any natural $d\ge 1$ and any field $\F$ of characteristic $0$ there exists a  polynomial $P\in \matd$ of degree $2d+1$ with $n$ variables  such that $Q_{S_{2d}}(P)=\Omega(n^{2d})$.
\end{lemma}

\begin{comment}
It can be shown that the lemma also holds for any finite \textit{field }$\F$. Since in Section \ref{sec:conc-for-any-basis-of-matd} we need to assume that the field is of characteristic $0$, we  prove the lemma only for fields of characteristic 0 .
\end{comment}

For proving the lemma, we introduce the following definition:

\begin{definition}\label{def:s-poly}
A polynomial $P\in\freea$ with $n$ variables $\nx{x}$ is called an \textbf{s-polynomial} if:
$$
P=\sum_{j_1<j_2<\ldots<j_{2d}\in [n]}c_{j_1j_2...j_{2d}}\cd S_{2d}\left(x_{j_1},                                                                                                                                                                                                x_{j_2}\ldots x_{j_{2d}}\right),
$$
for some natural $d$ and constants
$c_{j_1j_2....j_{2d}}\in \set{0,1}$, for
\(j_1<j_2<\ldots<j_{2d}\in [n]\).
\end{definition}

\begin{lemma} \label{fac:S_2d-equal-zero-when-constant}
For  any $P_1,P_2,\ldots,P_{2d}\in\freea$ where $d$ is a positive integer,  $S_{2d}(P_1,P_2,\ldots,P_{2d})$ is the zero polynomial if there exists $i\in[2d]$ such that \ $P_i$ is a constant.
\end{lemma}
\begin{proof}
For  a fixed $\ii\in[2d]$, we have  $P_\ii=c\in \F$.

For convenience, write the set $\set{x\in [2d]|x\neq\ii}$ as $[2d]/\ii$, the permutation $\left(\begin{array}{cccccccc}
1 &2 &\ldots& m-1&m&m+1&\ldots&2d\\                                                                                                                                                                                                                                                                                                                                i_1 &i_2 & \ldots&i_{m-1}&\ii  &i_m&\ldots&i_{2d-1}
\end{array}\right)$ as $\s_m$ where $\set{i_1,\ldots, i_{2d-1}}=[2d]/\ii$.

 Then
\begin{align*}
  S_{2d}(P_1,P_2,\ldots,P_{2d})=&\sum_{\sigma\in \S_{2d}}sgn(\sigma)\prod_{i=1}^{2d}P_{\sigma(i)}\\
  =& \prod_{\set{i_1,i_2,\ldots,i_{2d-1}}=[2d]/\ii} \sum_{m=1}^{2d} sgn(\s_m)\prod_{j=1}^{m-1}P_{i_j}P_\ii\prod_{j=m}^{2d-1}P_{i_j}\\
 =& \prod_{\set{i_1,i_2,\ldots,i_{2d-1}}=[2d]/\ii }\sum_{m=1}^{2d} sgn(\s_m)c \prod_{j=1}^{2d-1}P_{i_j}\\
    =&c\prod_{\set{i_1,i_2,\ldots,i_{2d-1}}=[2d]/\ii} \left(\sum_{m=1}^{2d} sgn(\s_m)\right) \prod_{j=1}^{2d-1}P_{i_j}\\
        =&c\prod_{\set{i_1,i_2,\ldots,i_{2d-1}}=[2d]/\ii} \left(\sum_{m=1}^{d}( sgn(\s_{2m-1})+sgn(\s_{2m}))\right) \prod_{j=1}^{2d-1}P_{i_j}\\
        =&c\prod_{\set{i_1,i_2,\ldots,i_{2d-1}}=[2d]/\ii} \left(\sum_{m=1}^d0\right) \prod_{j=1}^{2d-1}P_{i_j}\\
    =&0.
\end{align*}
\end{proof}

Any  s-polynomial  has the following property:
\iddotodo{Write ideals with $\langle ... \rangle $ notation and not $I(...)$.}
\iddofix{In the following lemma: what is $r$? Is it $n$? Also, don't use the parameter $n$ here, because it's already used by number of variables $ x_1,\ldots,x_n $.}
\begin{lemma}\label{lem:s-poly-linear-property}
Let $f$ be an s-polynomial. If there exist vectors of polynomials  $\overline{P_1},\ldots,\overline{P_r}$ with
$$f\in \ideal{S_{2d}(\overline{P_1}),\ldots,
S_{2d}(\overline{P_r})},$$
then
$$f=\sum_{i=1}^r c_iS_{2d}\left(\degr{\overline{P_i}}{1}\right).$$
\end{lemma}

\begin{proof}
Notice that the s-formula $f$ is $2d-$homogenous.
Thus,
$$f=\degr{f}{2d}\in \left\{\degr{h}{2d} \;\big| \; h\in \ideal{S_{2d}(\overline{P_1}),\ldots, S_{2d}(\overline{P_r})}\right\}.$$
That is

$$f\in \ideal{S_{2d}(\overline{P_{1}})^{(2d)},\ldots, S_{2d}(\overline{P_{r}})^{(2d)}}.$$

By Lemma \ref{fac:S_2d-equal-zero-when-constant},  for some $j\in[r],i\in[2d]$, the polynomial $S_{2d}(\t P_j)$  equals to the zero polynomial if some  $\t P_{j_i}$ is a constant.
Namely
$S_{2d}(\overline{P_{j}})^{(2d)}=S_{2d}
        \left(\degr{\overline{P_j}}{1}\right), \text{
for all }  j\in[r].$ Then,
$$f\in \ideal{S_{2d}\left(\degr{\overline{P_1}}{1}\right),\ldots, S_{2d}\left(\degr{\overline{P_r}}{1}\right)}.$$
That is,
\begin{align*}
f &=\sum_{j=1}^r \sum_{i=1}^{t_j} A_{ji}S_{2d}\left(\degr{\overline{P_j}}{1}\right)B_{ji}, ~~~~\text{for some  $A_{ji},B_{ji}\in\freea$.}
\end{align*}
Moreover,
$$
    \degr{A_{ji}S_{2d}
        \left(
            \degr{
                \overline{P_j}
                }{1}
        \right)B_{ji}
        }{2d}=
        \degr{A_{ji}B_{ji}}{0}S_{2d}
            \left(
                \degr{\overline{P_j}}{1}
            \right).
$$

Thus   $$f=\sum_{j=1}^r c_jS_{2d}\left(\degr{\overline{P_j}}{1}\right),$$
where $c_j$ is the constant  $\sum_{i=1}^{t_j}\degr{A_{ji}B_{ji}}{0}$, for any $j\in[r]$.
\end{proof}

\subsubsection{The counting argument}\label{sec:count-arg}

\begin{notation*}
If $B\subseteq \freea$ contains only one polynomial $g$, then we write $Q_g(\cd)$ instead of $Q_{B}(\cd)$, to simplify the writing. Note that $B$ may not be a basis for the algebra considered (e.g., we may consider identities of the \matd\ generated by some $B$, where $B$ is not a basis for (all) the identities of \matd).\iddotodo{This should come in Background?}
\end{notation*}

\begin{lemma}\label{lem:exist_for_nP}
For any field \F\ of characteristic $0$, there exist s-polynomials $\nx{P}$ which are identities of \matd\ in $n$ variables, such that $Q_{S_{2d}}(\nx{P})=\Omega (n^{2d})$ (and $Q_{S_{2d}}(\nx{P})$ is finite). \iddofix{***CHECK!***}
\end{lemma}

In Section \ref{sec:conc-for-any-basis-of-matd} we show that, if $\F$ is of characteristic $0$ then  this lower bound holds for \emph{any finite basis of} \matd, namely for $Q_B$, where $B$ is any finite basis of \matd.
\begin{proof}
We prove by a generalization of the counting argument from \cite{Hru11} that there exists a sequence of polynomials $P_1,P_2,\ldots,P_n$  that require  $\Omega \left(n^{2d}\right)$ substitution instances of the $S_{2d}(x_1,\ldots,x_{2d})$ identities to generate (all of the polynomials in the sequence) in a two-sided ideal.

Recall that an s-polynomial (Definition \ref{def:s-poly}) is of the following form:

\begin{equation}\label{eq:Pi-something}
\sum_{j_1<j_2<\ldots<j_{2d}\in [n]}c_{i_{j_1j_2\cdots j_{2d}}}S_{2d}(x_{j_1},x_{j_2},\ldots,x_{j_{2d}}),~\hbox{ where }c_{i_{j_1j_2\cdots j_{2d}}}\in \set{0,1}.
\end{equation}
\smallskip

Assume that

$$\l=\max \left\{Q_{S_{2d}}(\nx{P}) \;:\; \text{ $P_i$ is an s-polynomial, for all $i\in[n]$}\right\}.$$

 Then for any choice of $n$ s-polynomials $P_1,\ldots,P_n$ there are $\l$ vectors of polynomials $\overline{Q_1},\ldots,\overline{Q_{\l}}$ from \freea, such that
$$\nx{P}\in \ideal{S_{2d}(\overline{Q_1}),\ldots, S_{2d}(\overline{Q_{\l}})}.$$
By Lemma \ref{lem:s-poly-linear-property},
for any choice of $P_1,\ldots,P_n$ and $\overline Q_1,\ldots,\overline Q_\l$, for every $i\in[n]$:
\begin{align*}
P_i=\sum_{j=1}^{\l} c_{i_j}S_{2d}\left(\overline{Q_{j}}^{(1)}\right)&=\sum_{j=1}^{\l} c_{i_j}S_{2d}\left(\sum_{m=1}^n a_{mj_1}x_m,\sum_{m=1}^n a_{mj_2}x_m,\ldots,\sum_{m=1}^n a_{mj_{2d}}x_m\right)\\
&~~~~~~~~~~~~~~~~~~~~~~~~~~~~~~~~~~~~~~~(\text{for some } c_{i_j},a_{mj_{k}}\in \F).
\end{align*}

Consider a vector $(c_{1_j} ,\ldots, c_{n_j},a_{k1m},\ldots,
a_{k(2d)m})\;(m\in [n],k\in[\l]).$ By linearity of $S_{2d}$:
\begin{align}\label{eq:what-induces-map}
\sum_{k=1}^{\l}
    c_{i_k} S_{2d}
        \left(
            \sum_{m=1}^n
                a_{k1m} x_m,
                    \sum_{m=1}^n a_{k2m} x_m,\ldots,\sum_{m=1}^n a_{k(2d)m} x_m
        \right) = \text{ }
\\
\sum_{j_1<j_2<\ldots<j_{2d}\in[n]}
    c_{i_{j_1j_2\cdots j_{2d}}}
        S_{2d}(x_{j_1},x_{j_2},\ldots,x_{j_{2d}})
~~~~~~~~(\text{where }   c_{i_{j_1j_2\cdots j_{2d}}}\in \F).
\end{align}
A  \textit{polynomial map} $\mu : \F^n \rightarrow \F^m$ of degree $d > 0$, is a map
$\mu = (\mu_1, \ldots, \mu_m)$, where each $\mu_i$ is a (commutative) polynomial of degree $d$ with $n$ variables.

\begin{claim*}
Consider the coefficients $c_{1_j} ,\ldots, c_{n_j},a_{k1m},\ldots, a_{k(2d)m}$ and the coefficients $ c_{i_{j_1j_2\cdots j_{2d}}} $ in Equation \ref{eq:what-induces-map} \emph{as variables}. Then, Equation \ref{eq:what-induces-map} defines a degree-$(2d+1)$ polynomial map $\phi:\F^{({2d}+1)nl}\to \F^{n{ n\choose {2d}}}$ that maps each vector
\[
(c_{1_j} ,\ldots, c_{n_j},a_{k1m},\ldots,a_{k(2d)m}),~~\hbox{ for } m\in [n],k\in[\l],
\]
to
\[
(c_{1_{j_1j_2\cdots j_{2d}}} ,\ldots, c_{n_{j_1j_2\cdots j_{2d}}}),\hbox{ ~~for~} j_1<j_2<\ldots<j_{2d}\in[n].
\]
\end{claim*}
We omit the details of the proof of this claim.
We have the following lemma:
\begin{lemma}[\cite{HY11}, Lemma 5]
\iddofixl{Say that the polynomials are commutative here.}
For any field $\F$, if  $\mu$ : $\F^n \rightarrow \F^m$ is a  polynomial map of degree $d > 0$,  then $\left|\mu(\F^n) \bigcap\{0, 1\}^m\right| \leq  (2d)^n$.
\end{lemma}

Thus, for the degree-$(2d+1)$ polynomial map $\phi:\F^{({2d}+1)nl}\to \F^{n{ n\choose {2d}}}$, we have
 $$|\phi(\F^{({2d}+1)nl})\bigcap\set{0,1}^{n{ n\choose {2d}}}|\leq  (2(2d+1))^{({2d}+1)nl}.$$
\smallskip
%
%

Recall that for any choice of $n$ s-polynomials $P_1,\ldots,P_n$ there are $\l$ vectors of polynomials $\overline{Q_1},\ldots,\overline{  Q_{\l}}$ from \freea, such that
$$\nx{P}\in \ideal{S_{2d}(\overline{Q_1}),\ldots, S_{2d}(\overline{Q_{\l}})}.$$

For convenience, we use  $\t{\mathcal C}$ for the $0-1$ vector $(c_{1_{j_1j_2\cdots j_{2d}}} ,\ldots, c_{n_{j_1j_2\cdots j_{2d}}})$, where $c_{i_{j_1j_2\cdots j_{2d}}}\in \set{0,1},i\in[n],j_1<j_2<\ldots<j_{2d}\in[n]$.
Since for every possible $\t{\mathcal C}$, the following polynomials are s-polynomials:
$$\sum_{j_1<j_2<\ldots<j_{2d}\in [n]}\mathcal C_{1_{j_1 j_2\cdots j_{2d}}} S_{2d}(x_{j_1},x_{j_2},\ldots,x_{j_{2d}}),
~~~~\ldots,~~~\sum_{j_1<j_2<\ldots<j_{2d}\in [n]}\mathcal C_{n_{j_1 j_2\cdots j_{2d}}} S_{2d}(x_{j_1},x_{j_2},\ldots,x_{j_{2d}}),$$

 there exist $\l$ vectors of polynomials $\overline{Q_1},\ldots,\overline{Q_{\l}}$ in \freea, such that
$$\sum_{j_1<j_2<\ldots<j_{2d}\in [n]}\mathcal C_{i_{j_1 j_2\cdots j_{2d}}} S_{2d}(x_{j_1},x_{j_2},\ldots,x_{j_{2d}})\in \ideal{S_{2d}(\overline{Q_1}),\ldots, S_{2d}(\overline{Q_{\l}})},i\in[n].$$
That is,  there exists a vector $(c_{1_j} ,\ldots, c_{n_j},a_{k1m},\ldots,a_{k(2d)m})\;(m\in [n],k\in[\l]),$ such that $\phi(c_{1_j} ,\ldots, c_{n_j},a_{k1m},\ldots,a_{k(2d)m})=\t{\mathcal C}$.

\smallskip
Therefore,  every possible $\t{\mathcal C}$ belongs to $\phi(\F^{({2d}+1)nl})\bigcap\set{0,1}^{n{ n\choose {2d}}}.$

Further there are $2^{n{ n\choose {2d}}}$ distinct vectors $\t{\mathcal C}=(c_{1_{j_1j_2\cdots j_{2d}}} ,\ldots, c_{n_{j_1j_2\cdots j_{2d}}})$, where $c_{i_{j_1j_2\cdots j_{2d}}}\in \set{0,1},i\in[n],j_1<,\ldots, <j_{2d}\in[n]$.
Hence,
$$|\phi(\F^{({2d}+1)nl})\bigcap\set{0,1}^{n{ n\choose {2d}}}|\ge 2^{n{ n\choose {2d}}}.$$

This implies that
\begin{equation}
(2(2d+1))^{({2d}+1)nl}\ge 2^{n{ n\choose {2d}}}.
\end{equation}
Using the $\ln$ function on both sides:
$$(2d+1)nl\ln (2(2d+1)) \ge n {n\choose{2d}}\ln 2.$$
Hence,
\begin{equation}\label{eq:count_constant_of_dimension_1}
  l> \frac{{n\choose {2d}}\ln 2}{(2d+1)\ln (4d+2)}.
\end{equation}
Namely
\begin{align*}
l&> c{n\choose {2d}}=c\frac{n(n-1)\dots(n-2d+1)}{d!}=\Omega\left(n^{2d}\right)\\
&~~~~~~~~~~~~~~~~~~(\text{where }c\in \F), \hbox{  hence}
\end{align*}
$$l=\Omega\left(n^{2d}\right).$$
\end{proof}

\subsubsection{Combining the polynomials into one}

Here we show that there exists already a \textit{single}
polynomial, denoted $\dot{P}$   such that $Q_{S_{2d}}(\dot{P})=\Omega(n^{2d})$. This is done in a manner which is  similar to the work of Hrube\v s \cite{Hru11}; however, there is a further complication here, which is dealt via the technical Lemma \ref{lem:transfer-polynomials}.

 \begin{lemma}
\label{lem:combine_into_one}
Let $\nx{P}$ be s-polynomials in n variables $\nx{x}$, and  let $\nx{z}$ be new variables,  different from $\nx{x}$. Let $\dot{P}$:=$\sum_{i=1}^n z_iP_i$. Then
\begin{equation}\label{eq:count_constant_of_dimension_2}
  Q_{S_{2d}}(\dot{P})\geq\frac{1}{2d+1}Q_{S_{2d}}(\nx{P}).
\end{equation}
Specifically, for any field \F\ of characteristic 0 and every $d\ge 1$, there exists a polynomial with $n$ variables such that
$Q_{S_{2d}}(\dot{P})=\Omega(n^{2d})$.
\end{lemma}

\begin{proof}
For convenience, call the new variables $\nx{z}$  the $Z$-variables. Given a polynomial $f$, the  \textbf{\emph{$Z$-homogenous part of degree $j$ of $f$}}, denoted $\zd{f}{(j)}$, is the sum of all monomials where the total degree  of the $Z$-variables is $j$. For example if $f=z_1xy+z_2z_1+z_3x+1+x$, then $\zd{f}{1}=z_1xy+z_3x$, $\zd{f}{2}=z_2z_1$, $\zd{f}{0}=1+x$.
A polynomial  that does not contain any $Z$-variable is said to be \emph{$Z$-independent}.

First, we claim the $\dot{P}$ has the following property:
\begin{claim*}
For any    $\l$ $Z$-independent polynomials $\overline  G_{1},\overline G_2,\ldots,\overline G_{\l}\in \freea $, if
$$\dot{P}\in \ideal{S_{2d}(\t G_{1}),\ldots,S_{2d}(\t G_{\l}) },$$
then
$$\nx{P}\in \ideal{  S_{2d}(\t G_{1}),\ldots,S_{2d}(\t G_{\l})}.$$
\end{claim*}
\begin{proofclaim}
Since $\dot{P}\in \ideal{S_{2d}(\t G_{1}),\ldots,S_{2d}(\t G_{\l})}$,
$$\dot{P}=\sum_{i=1}^n z_iP_i=\sum_{j=1}^\l \sum_{i=1}^{t_j} f_{ji}S_{2d}(\t G_{j})g_{ji}, ~~~~\text{for some  $f_{ji},g_{ji}\in\freea$}.$$
Now, assign   $z_1=1, z_2=z_3=\dots=z_n=0$ in $\dot{P}$. Since  $\t G_1,\ldots,\t G_{\l}$ do not  contain  $\nx{z}$, the $\t G_1,\ldots,\t G_{\l}$ will remain the same. Thus,
$$P_1=\sum_{j=1}^\l \sum_{i=1}^{t_j} f'_{ji}S_{2d}(\t G_{j})g'_{ji},$$
where $f'_{ji}=f_{ji}|_{z_1\leftarrow 1,z_2\leftarrow 0,\ldots,z_n\leftarrow 0},g_{ji}'=g_{ji}|_{z_1\leftarrow 1,z_2\leftarrow 0,\ldots,z_n\leftarrow 0}$.
Namely, $P_1\in \ideal{  S_{2d}(\t G_{1}),\ldots,S_{2d}(\t G_{\l})}$.

Similarly, we can show $P_2,\ldots,P_n\in \ideal{  S_{2d}(\t G_{1}),\ldots,S_{2d}(\t G_{\l})}$.
Therefore,  $$\nx{P}\in \ideal{  S_{2d}(\t G_{1}),\ldots,S_{2d}(\t G_{\l})}.$$
\end{proofclaim}

In the following, assume $Q_{S_{2d}}(\dot{P})=\l$.
 That is, there are   $k$ vectors of polynomials $\overline  G_{1},\overline G_2,\ldots,\overline G_{\l}$ such that
 $$\dot{P}\in \ideal{S_{2d}(\t G_{1}),\ldots,S_{2d}(\t G_{\l})}.$$
 Namely $$\dot{P}=\sum_{i=1}^n z_iP_i=\sum_{j=1}^\l \sum_{i=1}^{t_j} f_{ji}S_{2d}(\t G_{j})g_{ji}, ~~~~\text{for some  $f_{ji},g_{ji}\in\freea$}.$$
If we can find   $\Number$ $Z$-independent
 vector of polynomials $\overline  G_{1},\overline G_2,\ldots,\overline G_{\Number}$ such that
$$\dot{P}=\sum_{j=1}^\l \sum_{i=1}^{t_j} f_{ji}S_{2d}(\t G_{j})g_{ji}\in \ideal{S_{2d}(\t G_{1}),\ldots,S_{2d}(\t G_{\Number}) }.
  $$
then we can, by the above claim, show that
$$\nx{P}\in \ideal{  S_{2d}(\t G_{1}),\ldots,S_{2d}(\t G_{\Number})},$$
which is the conclusion we want to prove: $$Q_{S_{2d}}(\nx{P})\leq\Number.$$

Now, to find the $\Number$ $Z$-independent vectors of polynomials $\overline G_{1},\overline G_2,\ldots,\overline G_{\Number}$ which generate $\dot{P}$, let $\anbra{\cdot}$ be a map that maps a polynomial $P\in \freea$ to a polynomial $\anbra{P}$ that is defined by the following three properties:

\begin{enumerate}
\item The map $\anbra{\cdot}$ is linear, namely $\anbra{\alpha G+\beta H}=\alpha\anbra{G}+\beta\anbra{H}$ for any polynomials $G,H$ and $\alpha,\beta$ $\in\F$;
and\item Let  $M$ be a monomial whose  $Z$-homogenous part is of degree $1$.  Thus,  $M$  can be uniquely written as $M_1 z_{i}M_2,z_{i}\in Z,$ where $M_1,M_2$ are $Z$-independent.
Then $$\anbra{M}=\anbra{M_1 zM_2}=zM_2 M_1\,;~~~\hbox{and}$$
\item For a monomial $M$ whose  $Z$-homogenous part is not of degree $1$,  $\anbra{M}=0$.
\end{enumerate}

For convenience,  in what follows, given the polynomials $f,g$ and the vector of polynomials $\t H,$ we denote $\zd{f}{0},\zd{\t H}{0},\zd{g}{0}$ by $\mathcal{F}, \overline{\mathcal{ H}},\mathcal{G}$, respectively.

\begin{claim*}
For any  polynomials $f_1,g_1 ,\ldots,  f_k,g_k$ and vector of polynomials $\t H$ with variables $\nx{X},\nx{ z}$:
$$\anbra{ \sum_{i=1}^k f_i S_{2d}(\t H) g_i}\in\ideal{S_{2d}(\t{\mathcal{H}}),S_{2d}(\overline {\mathcal{H}}|_{ {\mathcal{H}}_j\leftarrow \sum_{i=1}^k\mathcal{G}_i\mathcal{F}_i})},~~~~\hbox{
for any }j\in[2d].$$
\end{claim*}
\begin{proofclaim}
Consider the following:
\begin{align*}
\anbra{\sum_{i=1}^k f_i S_{2d}(\t H) g_i}=&\anbra{\zd{\sum_{i=1}^k f_i S_{2d}(\t H) g_i}{1}}\text{~~by Property 3 of $[\cd]$}\\
=&\anbra{ \sum_{i=1}^k \zd{f_i}{1}  S_{2d}(\overline{\mathcal{ H}})\mathcal{G}_i+\sum_{i=1}^k \sum_{j=1}^{2d} \mathcal{F}_iS_{2d}\left(\overline {\mathcal{H}}|_{\mathcal{H}_j\leftarrow \zd{ H_j}{1}}\right)\mathcal{G}_i+\sum_{i=1}^k \mathcal{F}_i S_{2d}(\t {\mathcal{H}})\zd{g_i}{1}}        \\
\text{(by linearity of $[\cd]$)}~~~~ =&\sum_{i=1}^k\anbra{ \zd{f_i}{1}S_{2d}(\t {\mathcal{ H}})\mathcal{G}_i}+\sum_{j=1}^{2d} \anbra{\sum_{i=1}^k \mathcal{F}_i S_{2d}\left(\t {\mathcal{H}}|_{ {\mathcal{H}}_j\leftarrow \zd{ H_j}{1}}\right)\mathcal{ G}_i}+\sum_{i=1}^k\anbra{ \mathcal{F}_i S_{2d}(\t {\mathcal{ H}})\zd{g_i}{1}}.
\end{align*}

For every $i\in[n]$, assume $\zd{f_i}{1}=\sum_{i=1}^n \sum_jg_{ij}z_ih_{ij}$ where $g_{ij},h_{ij}$ are $Z$-independent polynomials and $\nx{z}$ are  $Z$-variables,
then
$$\anbra{\zd{f_i}{1}S_{2d}(\t {\mathcal{ H}})\mathcal{G}_i}=\anbra{\sum_{i=1}^n \sum_jg_{ij}z_ih_{ij}S_{2d}(\t {\mathcal{ H}})\mathcal{G}_i}=\sum_{i=1}^n \sum_jz_ih_{ij}S_{2d}(\t{\mathcal{ H}})\mathcal{G}_i g_{ij}\in \ideal{S_{2d}(\t {\mathcal{ H}})}$$
where the right most equality
stems from Property 2 of the map $[\cd]$. Similarly, for every $i\in [n]$, we can show  $$\anbra{\mathcal{F}_i S_{2d}(\t {\mathcal{ H}})\zd{g_i}{1}}\in \ideal{S_{2d}(\t {\mathcal{ H}})}.$$

By Lemma \ref{lem:transfer-polynomials}, which is proved below, we have  $$\anbra{\sum_{i=1}^k\mathcal{ F }_i S_{2d}(\t {\mathcal{H}}|_{\mathcal{H}_j\leftarrow \zd{ H_j}{1}})\mathcal{G}_i} \in  \ideal{S_{2d}(\overline {\mathcal{H}}|_{\mathcal{H}_j\leftarrow \sum_{i=1}^k\mathcal{G}_i\mathcal{F}_i})}, ~~~~\hbox{ for any } j\in[2d]. $$

Thus    $\anbra{\sum_{i=1}^kf_iS_{2d}(\t H)g_i} \in  \ideal{S_{2d}(\t H),S_{2d}(\overline {\mathcal{H}}|_{\mathcal{H}_j\leftarrow \sum_{i=1}^k\mathcal{G}_i\mathcal{F}_i})}$ for any $j\in[2d] $.
\end{proofclaim}

Note that  $\dot{P}=\zd{\dot{P}}{1}$. By the properties of $[\cd]$ we have:
\begin{align*}
\dot{P}&=\anbra{\dot{P}}\\
&=\anbra{\sum_{j=1}^\l \sum_{i=1}^{t_j} f_{ji}S_{2d}(\t H_{j})g_{ji}}\\
&=\sum_{j=1}^\l \anbra{ \sum_{i=1}^{t_j} f_{ji}S_{2d}(\t H_{j})g_{ji}}\\
&\in \ideal{S_{2d}(\t H_j),S_{2d}(\t H_j|_{H_{j_q}\leftarrow \sum_{m=1}^{t_j} \mathcal{G}_{jm}\mathcal{F}_{jm}})} \hbox{~~~for any }j\in[\l],q\in[2d].
\end{align*}

Namely for $\dot{P}=\sum_{j=1}^\l \sum_{i=1}^{t_j} f_{ji}S_{2d}(\t H_{j})g_{ji}$, we  have $\Number$  $Z$-independent polynomials  that generate $\dot{P}$,  concluding the  theorem.
\end{proof}

\begin{lemma}\label{lem:transfer-polynomials}
Let $X=\{x_1,x_2,\ldots,x_n\}$ and $f_1,g_1 ,\ldots,  f_k,g_k \in  \F \langle X \rangle $. Let $Z=\{z,z_1,z_2,\ldots,z_n\}$ and assume that  $n$ is an even positive integer, and let
$\t P $ be a vector of polynomials $(P_1,P_2,\ldots,P_n)$ with variable set $X\cup Z$. We denote $\zd{\t P}{0}$, $\zd{f_i}{0},
\zd{g_i}{0}$ by $\overline {\mathcal{P}},\mathcal F_i,\mathcal G_i,i\in[k]$, respectively. Then, for any $j\in [n]$, it holds that
$$\anbra{\sum_{i=1}^k\mathcal F_i S_{n}(\t {\mathcal P}|_{{\mathcal P}_j\leftarrow \zd{P_j}{1}})\mathcal G_i } \in  \ideal{ S_{n}(\t{\mathcal P}|_{{\mathcal P }_j\leftarrow \sum_{i=1}^k\mathcal G_i\mathcal F_i})}. $$
\end{lemma}
For example, when $n=2$, the above lemma shows the following:
$$\anbra{\sum_{i=1}^k\mathcal F_iS_2(\zd{P_1}{1},\mathcal P_2)\mathcal G_i}\in \ideal{S_2(\sum_{i=1}^k\mathcal G_i\mathcal F_i,P_2  )},$$
$$\anbra{\sum_{i=1}^k\mathcal F_i S_2(\mathcal P_1,\zd{P_2}{1})\mathcal G_i}\in \ideal{S_2(P_1, \sum_{i=1}^k\mathcal G_i\mathcal F_i )}.$$
\begin{proof}
For  a fixed ${\ii}\in [n]$, we have $\zd{ P_\ii}{1 }=\uzv$, where $z \in Z,\, \U,\V\in\freea$ and $\U,\V$ are $Z$-independent.

For a permutation $\s\in \S_n$ and the polynomial vector $\t P= (\nx{P})$, we let
$$(\t P)_{\s[i,j]}=\left\{
\begin{array}{ll}
\prod_{m=i}^j P_{\s(m)}, & i\leq j; \\
 1, & i>j.
\end{array}\right. $$

We write $\S_n/m$  to denote the set $\set{\s\in \S_n\;|\;\s(m)=\ii}$.

And define $$\pi_m = \left(\begin{array}{cccccccc}
1 &2 &...& n-m&n-m+1&n-m+2&...&n\\                                                                                                                                                                                                                                                                                                                                m+1 &m+2 & ...&n&m& 1 &...&m-1
\end{array}\right) \forall m\in [n]
.$$

\begin{fact}\label{fac:sign}
  $sgn(\pi_m)=(-1)^{m(n-m)+m-1}=(-1)^{nm-m(m-1)-1}=-1$.
\end{fact}
\begin{fact}\label{fac:perm-transfer}
  $ \t{P}_{\s[m+1,n]}\cd \t{P}_{\s[1,m-1]}=\t{P}_{\s\pi_m[1,n-m]}\cd \t{P}_{\s\pi_m[n-m+2,n]}$, for all $\s\in \S_n/m$.
\end{fact}

\begin{fact}\label{fac:set-transfer}
  $(\S_n/m)\pi_m= \S_n/(n-m+1)$.
\end{fact}
So we have the following:
\begin{align*}
&\anbra{\sum_{i=1}^k\mathcal {F}_i  s_{n}(\t{\mathcal P}|_{\mathcal P _\ii\leftarrow \uzv})\mathcal {G}_i }\\
=&\anbra{\sum_{i=1}^k \mathcal {F}_i   \sum_{\s \in \S_{n} }sgn(\s)(\t{\mathcal P}_{\s[1,n]})|_{ \mathcal P _\ii\leftarrow \uzv} \mathcal {G}_i }\\
=&\anbra{\sum_{i=1}^k \mathcal {F}_i  \sum_{m=1}^{n} \sum_{\begin{array}{c}
\s \in \S_{n} \\
\s^{-1}(i)=m
\end{array}}sgn(\sigma)(-1)^m(\t{\mathcal P}_{\s[1,m-1]}\mathcal P_{\s(m)} \t{\mathcal P}_{\s[m+1,n]})|_{\mathcal P _\ii\leftarrow \uzv}\mathcal {G}_i } \\
=&\anbra{\sum_{i=1}^k \mathcal {F}_i  \sum_{m=1}^{n} \sum_{\s\in \S_{n}/m} sgn(\s)(-1)^m(\t{\mathcal P}_{\s[1,m-1]} \mathcal P _\ii \t{\mathcal P}_{\s[m+1,n]})|_{\overline{ \mathcal P }_\ii\leftarrow \uzv}\mathcal {G}_i } \\
=&\anbra{\sum_{i=1}^k \mathcal {F}_i  \sum_{m=1}^{n} \sum_{\s\in \S_{n}/m}sgn(\s)(-1)^m(\t{\mathcal P}_{\s[1,m-1]}\uzv \t{\mathcal P}_{\s[m+1,n]})\mathcal {G}_i } \\
=& \sum_{i=1}^n\sum_j z_i \V\sum_{m=1}^{n} \sum_{\s\in \S_{n}/m}sgn(\s)(-1)^m \t{\mathcal P}_{\s[m+1,n]}\left(\sum_{i=1}^k\mathcal {G}_i \mathcal {F}_i\right) \t{\mathcal P}_{\s[1,m-1]} \U\\
=& \sum_{i=1}^n \sum_j z_i \V\sum_{m=1}^{n} \sum_{\s\in \S_{n}/m}sgn(\s)(-1)^m \t{\mathcal P}_{\s\pi_m[1,n-m]}\left(\sum_{i=1}^k\mathcal {G}_i \mathcal {F}_i\right) \t{\mathcal P}_{\s\pi_m[n-m+2,n]} \U~~~~~\text{by \textbf{Fact} \ref{fac:perm-transfer}}\\
=& \sum_{i=1}^n \sum_j z_i \V\sum_{m=1}^{n} \sum_{\s\in \S_{n}/m}sgn(\s\pi_m)sgn(\pi_m)(-1)^m \t{\mathcal P}_{\s\pi_m[1,n-m]}\left(\sum_{i=1}^k\mathcal {G}_i \mathcal {F}_i\right) \t{\mathcal P}_{\s\pi_m[n-m+2,n]} \U.\\
&\text{let $\pi=\s\pi_m$, then $\pi\pi_m^{-1}=\s$,}\\
=& \sum_{i=1}^n \sum_j z_i \V\sum_{m=1}^{n} \sum_{\pi\pi_m^{-1}\in \S_{n}/m}sgn(\pi)(-1)(-1)^m \t{\mathcal P}_{\pi[1,n-m]}\left(\sum_{i=1}^k\mathcal {G}_i \mathcal {F}_i\right) \t{\mathcal P}_{\pi[n-m+2,n]} \U~~~~\text{by \textbf{Fact} \ref{fac:sign}}\\
=&-\sum_{i=1}^n \sum_j z_i \V\sum_{m=1}^{n} \sum_{\pi\in \S_{n}/(n-m+1)}sgn(\pi)(-1)^m \t{\mathcal P}_{\pi[1,n-m]}\left(\sum_{i=1}^k\mathcal {G}_i \mathcal {F}_i\right) \t{\mathcal P}_{\pi[n-m+2,n]} \U~~~~~\text{by \textbf{Fact} \ref{fac:set-transfer}}\\
&\text{let $m'=n-m+1$, then $m=n-m'-1$,}\\
=& - \sum_{i=1}^n \sum_j z_i \V\sum_{m'=1}^{n} \sum_{\pi\in \S_{n}/m'}sgn(\pi)(-1)^{n-m'+1} \t{\mathcal P}_{\pi[1,m'-1]}\left(\sum_{i=1}^k\mathcal {G}_i \mathcal {F}_i\right) \t{\mathcal P}_{\pi[m'+1,n]} \U\\
=& -(-1)^{n+1} \sum_{i=1}^n \sum_j z_i \V\sum_{m'=1}^{n} \sum_{\pi\in \S_{n}/m'}sgn(\pi)(-1)^{m'} \t{\mathcal P}_{\pi[1,m'-1]}\left(\sum_{i=1}^k\mathcal {G}_i \mathcal {F}_i\right) \t{\mathcal P}_{\pi[m'+1,n]} \U\\
=&  \sum_{i=1}^n \sum_j z_i \V S_{n}(\t{\mathcal P}|_{ \mathcal P _\ii\leftarrow \sum_{i=1}^k\mathcal {G}_i \mathcal {F}_i})\U\\
\in&\ideal{S_{n}(\t{\mathcal P}|_{ \mathcal P _\ii\leftarrow \sum_{i=1}^k\mathcal {G}_i \mathcal {F}_i})}.
\end{align*}
\end{proof}

\subsubsection{Concluding the lower bound for every basis of the identities of \matd}
\label{sec:conc-for-any-basis-of-matd}


Here we show that the $\Omega(n^{2d})$ lower bound proved in previous sections holds (for every $d>2$ and) \emph{every  finite basis of the identities of \matd}, when \F\ is of characteristic $0$.
To this end, we use several results from the theory of PI-algebras (for more on PI-theory see the monographs \cite{Row80,Dre99}).

A polynomial $f\in \freea$ with $n$ variables is \textbf{\emph{multi-homogenous with degrees $(1,\ldots,1)$}} ($n$ times) if in every monomial the power of every variable $x_1,\ldots,x_n$ is precisely 1. In other words, every monomial is of the form $\alpha\cd \prod _{i=1}^n x_{\sigma(i)}$, for some permutation $\sigma$ of order $n$ and some scalar $\alpha$. For the sake of simplicity, we shall talk in the sequel about a \textit{\textbf{multi-homogenous polynomial of degree $n$}}, when referring to a multi-homogenous polynomial with degrees $(1,\ldots,1)$ ($n$ times). Thus, any multi-homogenous polynomial with $n$ variables is homogenous of total-degree  $n$.

We need the following definition:

\begin{definition}\label{def:commutator_identity}
A polynomial $f\in \freea$ is called a \textbf{commutator polynomial} if it is a linear combination of products of generalized-commutators.
(We assume that $1$ is a product of an empty set of commutators.) \end{definition}
For example, $[x_1,x_2]\cd[x_3,x_4]+[x_1,x_2,x_3]$ is a commutator polynomial.

We need the following proposition:
\begin{proposition}[Proposition $4.3.3$ in \cite{Dre99}]
\label{prop:generated-by-multi-commutator-polynomial}
   If $R$ is a unitary PI-algebra over a field $\F$ of characteristic $0$, then every identity of $R$ can be generated by  multi-homogenous commutator  polynomials.
\end{proposition}
\begin{remark*}
\emph{Multi-homogenous} and \emph{commutator polynomials}, in the current paper, are called \emph{multilinear} and \emph{proper polynomials} in \cite{Dre99}, respectively.
\end{remark*}

\begin{lemma}\label{lem:special-basis}
  Let $R$ be a  unitary PI-algebra  and let  $\mathcal T$ be the T-ideal consisting of all identities of $R$. Then $\mathcal T$ has a  finite basis in which every polynomial is a multi-homogenous commutator polynomial.
\end{lemma}

\begin{proof}
By Kemer \cite{Kem87}, the identities of any $\F$-algebra, for any \F, \iddocomment{Is it for any field?}  can be generated by a finite set of identities. Namely $\mathcal T$ has a finite basis $\{A_1, ,\ldots,  A_k\}$, for some positive integer $k$.

By Proposition \ref{prop:generated-by-multi-commutator-polynomial}, for a fixed identity of $R$, we can find finite many  multi-homogenous commutator  polynomials to generate.
Thus, each $A_i,\,i\in[k]$, can be generated by finite many  multi-homogenous commutator polynomials. Then there are finite many multi-homogenous commutator polynomials generating the basis $\{A_1, ,\ldots, A_k\}$ of $\mathcal T$, and hence, also finite many multi-homogenous commutator identities generating $\mathcal T$.

\end{proof}

\begin{lemma}\label{lem:collapse}
  Let  $f\in \freea$ be a  multi-homogenous  commutator polynomial with $n$ variables.
 If $x_i$ is a constant for some  $i\in[n]$,  then $f(\nx{x})\equiv 0$ (that is, $f$ is the zero polynomial).
\end{lemma}
\begin{proof} In the proof, when we talk about the commutator, we mean the non-zero polynomial $[x_{t_1},\ldots,x_{t_s}]$ for all possible $t_1,\ldots, t_s\in[n]$ and some natural number $s\geq 2$. It is easy to check that if we replace a variable by  a constant $c\in \F$ in the commutator $[x_{t_1},\ldots,x_{t_s}]$, then the commutator equals $0$.

By the definition of commutator polynomial, we know
$$f=\sum_{i=1}^m c_i\prod_{j=1}^{k_i} B_{ij},$$
$~~~\text{where }0\neq c_i\in \F\hbox{ and }m,n\in \N,\text{ and  $B_{ij}$ is some  commutator $[x_{i_1},\ldots,x_{i_s}]$}. $
\smallskip

For a fixed $\ii\in[n]$, by the  definition of multi-homogenous polynomial, $f$ must be linear in  $x_\ii$, namely $c_i\prod_{j=1}^{k_i}B_{ij}$ must be linear in  $x_\ii$ for every $i\in[m]$. Then there must be a $j_0\in[k]$ such that $B_{ij_0}$ is linear in $x_\ii$. That is, $B_{ij_0}|_{x_\ii\leftarrow c}=0$. Furthermore,  $\prod_{j=1}^{k_i}B_{ij}|_{x_\ii\leftarrow c}=0$ for all $i\in[m]$. Namely
$f|_{x_\ii\leftarrow c}=0$.
\end{proof}


By lemma \ref{lem:exist_for_nP} and lemma \ref{lem:combine_into_one}, we know that there exist s-polynomials $\nx{P}$ in $n$ variables $\nx{x}$ that are identities over  \matd, such that putting  $\dot{P}$:=$\sum_{i=1}^n z_iP_i$, where  $\nx{z}$ are new variables, we have:
$$Q_{S_{2d}}(\dot{P})\geq\frac{1}{2d+1}\cd Q_{S_{2d}}(\nx{P})=\Omega(n^{2d}).$$
\medskip

The following is the main lemma of this section:
\begin{lemma}\label{lem:relation-S-2d-Matd}
  Let $d>2$, and let $\mathcal B$ be some basis for the T-ideals of the identities of \matd. Then, there are constants $c,c'$ such that for any identity $P$ over \matd\ of   degree $2d+1$:

  $$  c Q_{S_{2d}}(P)\le Q_{\mathcal B}(P)\le c' Q_{S_{2d}}(P).$$

\end{lemma}

To prove this theorem we need
the following two lemmata.
\begin{lemma}\label{lem:generated-by-S_2d}
For any natural number $d>2$, every multi-homogenous identity  (with any number of variables) over \matd\ of degree at most $2d+1$ is a consequence of the standard identity $S_{2d}$.
\end{lemma}
\begin{proof}
By Leron \cite{Ler73}, we know that  for any  $d>2$ every multi-homogenous identity of \matd\ with degree $2d+1$  is a consequence of the standard identity $S_{2d}$. By Exercise $7.1.2$ in  \cite{Dre99}, there are no identities of degree less than $2d$ in \matd\ and  every multi-homogenous polynomial identity of degree $2d$  in  $\matd$ is also a consequence of the standard identity $S_{2d}$.
\end{proof}

By Lemma \ref{lem:special-basis}, there is a basis  $\{A_1,A_2,\ldots,A_{m}\}$ of $\matd$, where $A_1,\ldots,A_m$ are all multi-homogenous commutator identities (Definition \ref{def:commutator_identity}).

\begin{lemma}\label{lem:2d+2-cant-generate}
Let $P$  be an identity of \matd\ of degree $2d+1$ and let $G$ be a basis  $\{A_1,A_2,\ldots,A_m\}$ of $\matd$, where $A_1,\ldots,A_m$ are all multi-homogenous  commutator identities of \matd. And assume $Q_{G}(P)=k$, that is, $k$ is the minimal number such that exist $k$ substitution instances $B_1,B_2,\ldots,B_k$ of $A_1,A_2,\ldots,A_m,$ for which:
$$P\in \ideal{B_1,B_2,\ldots,B_k}.$$
Then, no $B_\l$, for $\l\in[k]$, is a substitution instance of a basis element $A_j$ whose degree  is  greater than $2d+1$.
\end{lemma}

\begin{proof}
Assume there is $A_j$ (for $j\in[m]$) in the basis $G$ such that  the degree of $A_j(\overline x)$ is greater than $2d+1$.
In the following, we show that none of  $B_\l$\,($\l\in[k]$) is  a substitution instance of $A_j$.

Assume otherwise. Hence, there is a $B_\ii,\,\ii\in[k]$, such that $B_\ii$  is the substitution instance of $A_j$.
Since $A_j(\overline x)$ is homogeneous, every term in $A_j(\overline x)$ is of degree greater than $2d+1$.

We consider  the following two cases:
\medskip

\case 1 Every term in the  $A_j(\overline Q)$, which is a  substitution instances of $A_j(\t x)$, is  of degree greater than $2d+1$.

For convenience, given a polynomial $f$, we denote by $f^{\leq j}$ the polynomial    $\sum_{i=0}^{j}\degr{f}{i}$, namely the sum of all homogenous parts of $f$ of degree at most  $j.$ We consider the $2d+1$ homogenous part, that is:
\begin{align*}
  P&=\degr{P}{2d+1}\\
  &\in \set{\degr{h}{2d+1}\;\big|\; h\in\ideal{B_1,B_2,\ldots,B_k}}\subset \ideal{\degr{B_1}{\leq 2d+1},\ldots,\degr{B_k}{ \leq 2d+1}} .
\end{align*}
But $\degr{B_\ii}{\leq 2d+1}=\degr{A_j(\overline Q)}{\leq 2d+1}=0$, because, in this case, every term in $A_j(\overline Q)$ is of degree greater than $2d+1$. So $P$ can also belong to the ideal generated by $\set{\degr{B_1}{\leq 2d+1},\degr{B_2}{\leq 2d+1},\ldots,
\degr{B_k}{\leq 2d+1}}\setminus\degr{B_\ii}{\leq 2d+1}$.
This  means $Q_{G}(P)=k-1$ which contradicts  $Q_G(P)=k$. Thus  the assumption is false.\medskip

\case 2 There is a term of degree at most  $2d+1$ in $A_j(\overline Q)$, which is a  substitution instance of $A_j(\t x)$.

But we assumed that every term in $A_j(\t x)$ must be of degree greater than $2d+1$.  This means one of the  coordinates  of $\overline Q$ must be a constant.
That is, $A_j(\overline Q)=0$ (by Lemma \ref{lem:collapse}).
So $P$ can be generated by $\set{B_1,B_2,\ldots,B_k}\setminus B_i$. Hence, $Q_{G}(P)=k-1,$ which contradicts  $Q_G(P)=k$. Thus  the assumption is false.

\medskip

Now we can conclude that the assumption that there is a $B_\ii,\; \ii \in[k],$ such that $B_\ii$ is a substitution instance of $A_j$ is false.
So none of $B_\l$\;($\l\in[k]$) is  a substitution instance of $A_j$.
\end{proof}

We are now  back to the proof of Lemma \ref{lem:relation-S-2d-Matd}:
\begin{proof}
Let $\mathcal B$ be a basis  $\{A_1,A_2,\ldots,A_m\}$ of $\matd$, where $A_1,\ldots,A_m$ are all multi-homogenous  commutator identities of \matd.
Let
$$
\degr{\mathcal B}{\leq 2d+1}:=\{A_i\in \mathcal B\;|\;\text{the degree of $A_i$ is no more than $2d+1$}\}.
$$
 For any identity $P$ of \matd\ of degree $2d+1$, by Lemma $\ref{lem:2d+2-cant-generate}$,
 $$Q_{\degr{\mathcal B}{\leq 2d+1}}(P)=Q_{\mathcal B}(P).$$
This also means that every identity of \matd \ of degree at most $2d+1$ can be generated by $\degr{\mathcal B}{\leq 2d+1}$. Thus,  $S_{2d}$ can be generated by $\degr{\mathcal B}{\leq 2d+1}$. Write $\degr{\mathcal B}{\leq 2d+1}$ as the set $\{A_1',A_2',\ldots, A_{m'}'\},\, m'\leq m$, where the degree of $A_i'$\;($\forall i\in[m']$) is less than $2d+1$. By Lemma \ref{lem:generated-by-S_2d}, $A_1',\ldots, A_{m'}$  is generated by $S_{2d}$.
Then, by Equation \ref{eq:propostion_generator_set}  in Proposition \ref{prop:generate-means-less-Q},  for any identity $P$ over \matd\ of   degree $2d+1$:   \begin{equation}\label{eq:count_constant_of_dimension_3}
  \frac{1}{Q_{\degr{\mathcal B}{\leq 2d+1}}(S_{2d})} Q_{S_{2d}}(P)\le Q_{\degr{\mathcal B}{\leq 2d+1}}(P)\le \left(\max_{B\in \mathcal B'}Q_{S_{2d}}(B)\right) Q_{S_{2d}}(P)~~~d>2.
\end{equation}

Namely, for every  identity  $P$  of \matd\ of degree  $2d+1$,there are constants $c,c'$ such that: $$c Q_{S_{2d}}(P)\le Q_{\mathcal B}(P)\le c' Q_{S_{2d}}(P)~~~~d>2.$$
\end{proof}

We can now conclude the main theorem of this section, Theorem \ref{thm:main_lower_bound}, which we restate for convenience:

\begin{main-lower-bound}
Let \F\ be any field of characteristic 0. For every natural number $d>2$ and for every finite basis \(\mathcal B\) of the T-ideal of identities of  \matd, there exists an identity \(P\) over \matd\ of degree $2d+1$ with $n$ variables, such that $Q_{\mathcal B}(P)=\Omega(n^{2d})$.
\end{main-lower-bound}

\para{Note on the case of $d=2$.}
When $d=2$,   Lemma \ref{lem:relation-S-2d-Matd} is not true. For example, the polynomial $f=[[x_1,x_2][x_3,x_4]+[x_3,x_4][x_1,x_2],x_5]$ is an identity over  ${\rm Mat}_2(\F),$ but in \cite{Ler73} it is proved that $f$ cannot be generated by $S_4$. Namely the restriction $d>2$ in Lemma \ref{lem:relation-S-2d-Matd}, and also in Theorem \ref{thm:main_lower_bound}, is essential for our proof.

\section{Relations to tensor-rank }
Here we show that in order to make the hard (non-explicit) instances
$f$ from Theorem \ref{thm:main_lower_bound} into explicit ones, means finding explicit tensors with high tensor-rank.
This generalizes (to any order) a similar observation made in \cite{Hru11} for order 3 tensors. This means that the\textit{ specific} hard instances we provide in Theorem \ref{thm:main_lower_bound} are not good candidates for proof complexity hardness, because it is reasonable to assume they do not have small size circuits.

\begin{definition}
A  tensor $A : [n]^r\to \F$ is a \textbf{\textit{simple tensor }}if there exist $r$ vectors $a_1,\ldots,a_r:[n]\to \F$ such that $A=a_1\otimes\cdots\otimes a_r$, where $\otimes $ denotes tensor product, that is, $A$ is defined by $A(i_1,i_2,\ldots, i_r)=
a_1(i_1)\cdots a_r(i_r)$.
\end{definition}
\begin{definition}
For  a tensor $A$, the \textbf{tensor rank} $rank(A)$ is the minimal $k$ such that there exist $k$ simple tensors $A_1,A_2,\ldots, A_k:[n]^r\to \F$  such that $A =\sum_{i=1}^kA_i $.
\end{definition}

\begin{definition}
For a natural number $n$, let $A$ be a tensor $[n]^{r+1}\rightarrow \F$. We define the \textbf{corresponding polynomials} (from \freea ) \textbf{\textit{of the tensor}} $A$ as follows:
$$f_{j_0}:=\sum_{j_1,j_2,\ldots,j_{r}\in [n]}A(j_0,j_1,\ldots, j_{r})S_{r}(x_{j_1},x_{j_2},\ldots,x_{j_{r}}),~~~\forall j_0\in[n].$$
\end{definition}

By the following theorem, if we find  an  collection  of \emph{explicit}\footnote{A polynomial is said to be \emph{explicit} if the coefficient of a monomial of degree $d$ is computable by algebraic circuits of size at most $\poly(d)$, where $d$ is a natural number.} s-polynomials  $\nx{f}$ over \matd\ such that $Q_{S_{2d}}(\nx{f})$ is $\Omega(n^{2d})$, then we can find an \textit{explicit}\footnote{A tensor $T:[n]^r\rightarrow \F$ is called explicit if $T(i_1 ,\ldots, i_r)$ can be computed by algebraic circuits of size at most polynomial in $\poly(r \lg n)$, that is, at most polynomial in the size of the input $(i_1 ,\ldots, i_r)$.\textbraceright} tensor $A:[n]^{2d+1}\rightarrow \set{0,1}$ with rank $\Omega(n^{2d})$, where the corresponding polynomials of A are the s-polynomials $\nx{f}$.

\begin{theorem}
For a natural number $n$, let $A_{\nx{f}}$ be a tensor $[n]^{r+1}\rightarrow \F$ and let $\nx{f}\in\freea$ be  the corresponding polynomials of $A_{\nx{f}}$, then:$$Q_{S_{2d}}(\nx{f})\leq rank(A_{\nx{f}}).$$
\end{theorem}

\begin{proof}

Assume $rank(A_{\nx{f}})=R$. Namely  we can find  $R$  simple tensors  $A_1,A_2,\ldots, A_R$ such that
\begin{equation}\label{eq:tensor-decomposition}
  A_{\nx{f}} = \sum_{i=1}^R A_i.
\end{equation}

For every $i\in[R]$,  by simple tensor's definition,  there exist $2d+1$ vectors $\a_0,\a_1,\ldots,\a_{2d}:[n]\to F$ such that $A_i=\a_0\otimes \a_1\otimes \dots \otimes \a_{2d}$. Namely  $A_i(i_0,i_1,i_2,\ldots, i_{2d})=\a_0(i_0)\a_1(i_1) \cdots \a_{2d}(i_{2d})$, where $i_0 ,\ldots,  i_{2d}\in[n]$.

Concerning the corresponding polynomials $ \nx{f}$ of $A_{\nx{f}}$, for every $j_0\in[n]$,
\begin{align*}
f_{j_0}&=\sum_{j_1,j_2,\ldots,j_{r}\in [n]}A_{\nx{f}}(j_0,\ldots,j_{2d})S_{2d}(x_{j_1},\ldots, x_{j_{2d}})\\
&=\sum_{j_1,j_2,\ldots,j_{r}\in [n]}\sum_{i=1}^RA_i(j_0,\ldots, j_{2d})S_{2d}(x_{j_1},\ldots, x_{j_{2d}})~~~\text{(by \ref{eq:tensor-decomposition})}\\
&=\sum_{i=1}^R\sum_{j_1,j_2,\ldots,j_{r}\in [n]}A_i(j_0,\ldots, j_{2d})S_{2d}(x_{j_1},\ldots, x_{j_{2d}})\\
&=\sum_{i=1}^R \a_{0}(j_0)\sum_{j_1,j_2,\ldots,j_{r}\in [n]}\a_1(j_1)\cdots \a_{2d}(j_{2d})S_{2d}(x_{j_1},x_{j_2},\ldots, x_{j_{2d}})\\
&=\sum_{i=1}^R \a_0(j_0)S_{2d}\left(\sum_{1\leq j\leq n}\a_1(j)x_j,\sum_{1\leq j\leq n}\a_2(j)x_j,\ldots,\sum_{1\leq j\leq n}\a_{2d}(j)x_j\right)\\
&=\sum_{i=1}^R \a_0(j_0)S_{2d}(\t P_i)\\
\end{align*}
(For convenience, write $\left(\sum_{1\leq j\leq n}\a_1(j)x_j,\sum_{1\leq j\leq n}\a_2(j)x_j,\ldots,\sum_{1\leq j\leq n}\a_{2d}(j)x_j\right)$ as $\t P_i$, for any $i\in[R]$).

Namely $$\nx{f}\in\ideal{S_{2d}\left(\t P_1\right),\ldots,S_{2d} \left(\t P_R\right)}.$$

Thus $ Q_{S_{2d}}(\nx{f})\leq R$, namely $Q_{S_{2d}}(\nx{f})\leq ~\text{rank}(A_{\nx{f}})$.
\end{proof}

By the above theorem, we have the following:
\begin{corollary}
  If there exists a n explicit  collection  of s-polynomials  $\nx{f}$ (that are all identities of)  \matd,\ such that $Q_{S_{2d}}(\nx{f}) = \Omega(n^{2d})$, then there exists an \textit{explicit} tensor $A:[n]^{2d+1}\rightarrow \set{0,1}$ with tensor-rank  $\Omega(n^{2d})$.
\end{corollary}

\section{Matrix identities as hard proof complexity candidates}\label{sec:app_rel_to_PC}
Here we seek to find connections between the work we have done above to the problem of proving lower bounds in proof
complexity.

Consider a matrix identity $f$ over \matd. It is a non-commutative polynomial. Let $f$ be a nonzero polynomial identity over \matd. Then $f$ is a nonzero non-commutative polynomial from \freea. If we substitute each (matrix) variable $x_i$ in $f$ by a \dbyd\ matrix of \emph{entry-variables} $\{x_{ijk}\}_{j,k\in[n]}$, then now $f$ corresponds to $d^2$ commutative zero polynomials, one for each entry computed by $f$. \iddofix{**phrase**} Accordingly, let $F$ be a non-commutative circuit computing $f$. Then under the above substitution of $d^2$ entry-variables to each variable in $F$, we get $d^2$  non-commutative circuits, each computing the zero polynomial \emph{when considered as commutative polynomials}.
 Formally, we define the set of $d^2$ non-commutative circuits corresponding to the non-commutative circuit $F$ as follows:

\begin{definition}[$\convert{F}$, $\convert{F=0}$]
\label{def:double-bracket}
 Let $F$ be a non-commutative circuit computing the polynomial $f\in\freea$, such that $f$ is an identity of \matd. We define
$\convert{F}$ as  the  set of $d^2$ circuits which are generated  from  bottom to top in the circuit of $F$ according to the following rules:
\begin{enumerate}

\item every variable x in $F$ corresponds to $d^2$ new variables $x_{ij},i,j\in[d]$;
\item every plus gate $X\oplus Y$, where $X,Y$ represent two circuits, in $F$ corresponds to $d^2$ plus gates $\oplus_{ij},i,j\in[d]$ where each plus gate $\oplus_{ij}$ connects the corresponding circuit $X_{ij}$ and $Y_{ij}$ which have been generated before; \iddocomment{"generated before.." is unclear}
\item every multiplication gate $X\otimes Y$ in $F$ corresponds to $d^2$ plus gates $\oplus_{ij},i,j\in[d]$ where each plus gate $\oplus_{ij} $  is connected to $d$ multiplication gates $\otimes_k,k\in[d]$ which represent the multiplication of two corresponding circuit $X_{ik}$ and $Y_{kj}$ that have been generated before. (Formally, plus gates have \emph{fan-in two}, and so $\oplus_{ij}$ is the root of a binary tree whose internal nodes are all plus gates and whose $d$ leaves are the product gates $\otimes_k$, $k\in[d]$.)

\end{enumerate}
We define $\convert{F=0}$ to be the set of equations between circuits, where each circuit in $\convert{F}$ equals the circuit $0$.
\end{definition}
\begin{fact}
Since every gate in $F$ corresponds  to at most $d^3$ gates in $\convert{F},$ we have:

$$\big|\convert{F}\big|=O\left(d^3 |F| \right)$$
(where $|F|$ denotes the size of $F$, that is  the number
of nodes in $F$ and $\big|\convert{F}\big|$ denotes the sum of size of all circuits in $\convert{F}$). Thus, if we fix the dimension of a matrix as a constant, then we can claim that $|\convert{f}|=\Theta(|f|)$.
\end{fact}

First, we recall the arithmetic proof system $\PC(\F)$ (introduced  in \cite{HT12}, and almost similarly in \cite{HT08}) for deriving (commutative) polynomial identities over a field  $\F$. The system manipulate arithmetic equations, that is, expressions of the form $F = G$ where  $F, G$ are circuits.

\begin{definition}[Arithmetic proofs \PC(\F)]\label{def:arithmetic_proofs}
Let \F\ be a field. The system $\PC(\F)$ proves equations of the form $F = G$, where $F, G$ are
non-commutative arithmetic circuits (over \F). The inference rules are:
\begin{align*}
&\frac{F = G}{G = F} \qquad\qquad \qquad\qquad \qquad
\frac{F = G \qquad G = H}
{F = H}\\
&\frac{ F_1= G_1\qquad F_2= G_2}{F_1+F_2= G_1+G_2} \qquad\qquad
\frac{F_1= G_1\qquad F_2= G2}{F_1 \times F_2= G_1 \times G_2}
\,.
\end{align*}
The axioms are equations of the following form, with $F, G, H$ ranging over non-commutative circuits:
\begin{align*}
&Identity: \qquad F=F\\
&Product \: commutativity:  \qquad F\cdot G= G\cdot F\\
&Addition\: commutativity: \qquad F+G=G+F\\
&Associativity:\qquad   F+(G+H)=(F+G)+H \qquad F\cdot(G\cdot H)=(F\cdot G)\cdot H\\
&Distributivity:\qquad F\cdot(G+H)=F\cdot G+F\cdot H\\
&Zero \:element: \qquad F+0=F \qquad F\cdot 0=0\\
&Unit \:element:\qquad F\cdot 1 =F \\
&Field\ ~identities: \qquad c=a+b \qquad d= a'\cdot b'\\
&~~~~~~~~~~~~~~~~~~~~~~\text{where $a,a',b,b',c,d\in \F$, such that the equations hold in $\F$.} \\
&Circuit \ axiom:\qquad F=F'
\text{~~~~if $F$ and $F'$ are (syntactically) identical when} \\
&~~~~~~~~~~~~~~~~~~~~~~~~~~~~~~~~~~~~~~~~~~\text{both are un-winded into \emph{formulas}.}
\end{align*}
Note that the Circuit axiom can be verified in polynomial time (see e.g., \cite{Jer04}).

A proof $\pi$ in $\PC(\F)$ is a sequence of equations $F_1= G_1, F_2= G_2, \ldots, F_k= G_k$, with $F_i, G_i$ circuits, such that every equation is either an axiom, or was obtained from previous equations by one of the derivation rules.
An equation $F_i= G_i$ appearing in a proof is also called a \textit{proof-line}. Denote by $|\vdash_{\PC(\F)}F|$  the minimum number of lines in a $\PC$ proof of   $F=0$. We say that $\pi$ is a \PC\ proof of \textit{a set} of equations if $\pi$ is a \PC\ and it contains all the equations in the set as proof-lines).
\end{definition}

For $\F$ an infinite field,  $f$ is an identity in \matd\ iff   $\convert{F=0}$ has a $\PC(\F)$ proof. This is easy to prove as follows: assume by contradiction
otherwise, then there must be an assignment $A$ that makes $g  \neq 0$. This follows since the field is infinite  (and so every non
zero polynomial has an assignment that does not nullifies the polynomial).
But this assignment $A$ (extended in any way to all entries) makes the
matrix identity nonzero, in contradiction to the assumption that
it is a matrix identity.

\begin{main-open}  Let $d$ be a positive natural number and let $ \mathcal B$ be a (finite) basis of the T-ideal of the identities of \matd. Assume that $f\in\freea$ is an identity over \matd, and let $F$ be a non-commutative algebraic circuit computing $f$. Then, the minimal number of lines in an arithmetic proof of the collection of  $d^2$ (entry-wise) equations $\llbracket F=0 \rrbracket_d$  corresponding to $F$ is lower bounded (up to a constant factor) in $Q_{\mathcal  B}(f)$. And in symbols:
$$\big|\vdash_{\PC(\F)}\convert{F=0}\big| = \Omega(Q_{\mathcal B}(f)).$$
\end{main-open}


\subsection{Conditions for exponential lower bounds}\label{sec:exponential-lower-bounds}
\todo{Make thinks here the same as in intro} Can we, even potentially,  obtain exponential lower bounds on $\PC(\F)$ proof size using the measure $Q_B(\cd)$ and assuming Conjecture 1 holds?
  The answer is yes, under certain further technical assumptions. We write the assumptions formally:


\medskip

\ind\textbf{Assumptions:}

\begin{enumerate}

\item \textbf{Refinement of Conjecture II}: Assume that for any $d$ and any basis $\mathcal B_d$ of the identities of \matd\ the number of lines in any $\PC(\F)$ proof of $\llbracket F= 0\rrbracket_d$ is at least  $\mathcal C_{\mathcal B_d}\cd Q_{\mathcal B_d}(f)$, where $\mathcal C_{\mathcal B_d}$ is a number depending on \(\mathcal B_{d}\)  and $F$ is the non-commutative arithmetic circuit computing $f$ (this is the same as Conjecture 1 except that now $\mathcal C_{\mathcal B_d}$ is not a constant).

\item Assume that for any sufficiently large $d$ and any basis $\mathcal B_d$ of the identities of \matd, there exists a number $c_{\mathcal B_d}$ such that for all sufficiently large \(n\) there exists an identity $f_{n,d}$ with  $Q_{\mathcal B_d}(f_{n,d})\ge c_{\mathcal B_d}\cd n^{2d}$. (The existence of such identities are known from our unconditional lower bound.)

\item Assume that for the \(c_{\mathcal B_{d}}\) in item 2 above:   $c_{\mathcal B_d}\cd \mathcal C_{\mathcal B_d}= \Omega\left(\frac{1}{\poly(d)}\right)$.

\item \textbf{(Variant of) Conjecture I}: Assume that the non-commutative arithmetic circuit size of $f_{n,d}$ is  at most $\poly(n,d)$.
\end{enumerate}

\ind\textbf{Corollary (assuming Assumptions 1-4 above)}:
There exists a polynomial size (in $n$) family of identities between non-commutative arithmetic circuits, for which any \PC\ proof requires exponential $2^{\Omega(n)}$ number of proof-lines.

\begin{proof} By the assumptions, every $\PC(\F)$-proof of $\llbracket f_{n,d}=0 \rrbracket_d$ has size at least $c_{\mathcal B_d}\cd
\mathcal C_{\mathcal B_d}\cd n^{2d}$.  Consider the family $\{ f_{n,d}\}_{n=1}^\infty$, \textit{where $d$ is a function of $n$}, and we take  $d=n/4$. Then, we get the following lower bound on the number of lines in   $\PC(\F)$-proofs of the family $\{ f_{n,d}\}_{n=1}^\infty$:
\[c_{\mathcal B_d}\cd \mathcal C_{\mathcal B_d}\cd n^{2d}=\frac{1}{\poly(n/4)}n^{n/2}= 2^{\Omega(n)}
,
 \]
which (by Assumption 4) is \textit{exponential} in the arithmetic circuit-size of the identities \(f_{n,d}\) proved.
\end{proof}
\para{Justification of assumptions.}

We wish to justify to a certain extent the new Assumptions 3 above (which lets us obtain the exponential lower bound). We shall use the s-polynomials for this. First, note that Assumption 2 holds for the case of the  s-polynomials, by Theorem \ref{thm:main_lower_bound}.

 We now show that the function $c_{B_d}$ does not decrease too fast. By Equations \ref{eq:count_constant_of_dimension_1}, \ref{eq:count_constant_of_dimension_2}
and \ref{eq:count_constant_of_dimension_3} in
Section \ref{sec:the-lower-bound}, we know that for any natural number $d$, there is an s-polynomial $f,$ such that:
$$Q_{B_d}(f)\geq \frac{1}{Q_{\degr{ B_d}{\leq 2d+1}}(S_{2d})}\frac{1}{2d+1}\frac{{n\choose {2d}}\ln 2}{(2d+1)\ln (4d+2)}.$$
Let $\mathcal B_{d}$ be a set of identities of \matd\  that contains the $S_{2d}$ identities.
Hence, $$Q_{ \degr{ \mathcal B_{d}}{ \leq 2d+1 }} (S_{2d})=1.$$
Thus
$$Q_{\mathcal B_d}(f)\geq \frac{1}{2d+1}\frac{{n\choose {2d}}\ln 2}{(2d+1)\ln (4d+2)}.$$

If we let $d=n/4$, then
$$Q_{\mathcal B_{n/4}}(f)\geq \frac{1}{n/2+1}\frac{{n\choose {n/2}}\ln 2}{(n/2+1)\ln (n+2)}.$$
By Stirling's formula, we get that $n!\sim\sqrt{2\pi n}(\frac{n}{e})^n$. Hence,
${n\choose n/2}\sim \frac{2^{n+1/2}}{\sqrt{n\pi}}$. Then $$Q_{\mathcal B_{n/4}}(f)=\Omega\left(\frac{2^n}{n^{5/2}\ln n}\right).$$
This shows that the function \(c_{B_{d}}\) does not decrease too fast.

\bigskip

We can use the fact that \(c_{B_{d}}\) does not decrease too fast to get the following (conditional exponential lower bound):


\begin{proposition}
Suppose \emph{Assumption 1} above holds (refinement of Conjecture 1) and assume that $\mathcal C_{\mathcal B_{n/4}}=
\Omega(1/{\rm poly}(n))$. Then, there exists a  family of
non-commutative circuits $\{F_n\}_{n=1}^\infty$ (computing the family of polynomials
$\{f_{n,\frac{n}{4}}\}_{n=1}^\infty$)
such that  the number of lines in any $\PC(\F)$-proof of $\llbracket F_n= 0\rrbracket_{n/4}$ is at least
$\mathcal C_{B_{n/4}}\Omega\left(\frac{2^n}{n^{5/2}\ln n}\right)=\Omega
\left(\frac{2^n}{\poly(n)}\right)=2^{\Omega(n)}$.
\end{proposition} \iddofix{check this ln in denominator}

Note that we get only an exponential lower bound in \(n\) for the lines of proofs in $\PC$ in the above consequence. But this does not entail an exponential lower bound in the size of  $\llbracket F_n= 0\rrbracket_{n/4}$ (the latter is polynomial in the size of the circuit \(F_{n}\), computing the s-polynomials. So this proposition is presented here in order to show that at least for some identities, the additional requirement (Assumption 3) on parameters, added to get a conditional  exponential lower bound, is attainable.

\subsection{A propositional version of Conjecture II}\label{sec:prop-vers}
We wish to comment on  the applicability of our suggested framework, for achieving propositional Extended Frege lower bounds.

It seems that the most natural way to connect the complexity, measure  \(Q_{\mathcal B}(\cd)\) to the number of
lines in an Extended Frege  (see, e.g., \cite{Kra95} or \cite{Jer04} for a formal definition of Extended Frege) proof is to require that the Main Open Problem states an \textit{even stronger} statement. Admittedly, this makes the new assumption, shown below,  quite speculative at the moment.

Given a commutative algebraic circuit $C$ over $GF(2)$, we can think of the circuit equation $C=0$ as a \emph{Boolean} circuit computing a tautology, instead of an algebraic circuit: interpreting $+$ as XOR, $\cd$ as $\land$, and $=$ as logical equivalence $\equiv$ (that is, $\leftrightarrow$).    Accordingly, we can consider arithmetic proofs over $GF(2)$ augmented with the Boolean axioms $x^2_i+x_i=0$, for each variables $x_i$, to obtain a propositional proof system which formally \textit{is} an Extended Frege proof system (see \cite{HT12}). Denote this system
$\PC(\F)+\{x_i^2+x_i=0\;:\;x_i\in X\}$.

Then, there is no clear reason to rule out the following:

\newtheorem*{prop-main-open}{Conjecture 1 for the propositional case over $GF(2)$}

\begin{prop-main-open} Let $\F=GF(2)$, let $d$ be a positive natural number and let $\mathcal B$ be a (finite) basis of the  identities of \matd. Assume that $f\in\freea$ is an identity of \matd, and let $F$ be a non-commutative algebraic circuit computing $f$. Then, the minimal number of lines in a $\PC(\F)+\{x_i^2+x_i\;=0:\;x_i\in X\}$ proof of the collection of  $d^2$ (entry-wise) equations $\llbracket F=0 \rrbracket_d$  corresponding to $F$ is lower bounded (up to a constant factor) by $Q_{\mathcal  B}(f)$. And in symbols:
\begin{equation}\label{eq:main-open-prob-prop}
\big|\vdash_{\PC(\F)+\{x_i^2+x_i=0\;:\;x_i\in X\}}\convert{F=0}\big| = \Omega(Q_{\mathcal B}(f)).
\end{equation}
\end{prop-main-open}
(Where, as before,
$\big|\vdash_{\PC(\F)+\{x_i^2+x_i\;=0:\;x_i\in X\}}\convert{F=0}\big|$ is the minimal size of a $\PC(\F)+\{x_i^2+x_i=0\;:\;x_i\in X\}$ proof containing all the
circuit-equations in $\convert{F=0}$.)

\medskip

\begin{comment} One can plausibly consider the same propositional version of the main open problem, with $\F$ being the rational numbers, and hence of characteristic $0$ (for we which we have more knowledge about $Q_{\mathcal B}(\cd)$, as obtained in our work). However, the way to translate arithmetic proofs \PC\ over the rationals is less immediate than the same translation for the case of $GF(2)$, and we have not verified formally the details of such a translation.
\end{comment}

\iddofix{\textbf{Some discussion:
}\\
Note that \PMd\ might be a \emph{strictly} decreasing hierarchy: because for $d=1$ we can always use only $O(n^2)$ number of commutativity
axioms, while we need more for bigger $d$'s. Nevertheless, why does this mean that there is a separation?
\\
Nevertheless, whether it is or not a proper hierarchy, we might ask: if the hierarchy is proper, why should we expect that the proof system \PC, which is the strongest in the hierarchy, would actually ``inherit'' the lower bound of \PMd\ for higher $d$'s?
\\
The answer we give is that it inherits the hardness of \PMd\ for higher $d$'s, because we use the translation $\llbracket \cd \rrbracket$; and the Main Open Problem states that this translation is linked to $Q_{\mathcal B}(\cd)$.
}

%
%
%

\subsection{Hierarchy of proofs for matrix identities }\label{sec:system-diff-algebras}
The proof system $\PC(\F)$ works for proving identities over commutative fields. Here we  formulate  a fragment of \PC(\F) that proves matrix \matd\ identities, for every given $d$. In what follows, $\F$ always denotes a field of characteristic $0$.

For any field \F\ (of characteristic 0), any $d\ge 1,$ and any basis \(\mathcal B\) of the identities of \matd, we define the following proof system \PMd, which is sound and complete for the identities of \matd\ (written as equations of non-commutative circuits): consider the proof systems  \PC(\F)\ and \textit{replace }the commutativity axiom \(h\cd g=g\cd h\) by  a finite basis \(\mathcal B\) of the identities of  \matd\  (namely, add a new axiom $H=0$ for each polynomial $h$ in the basis, where $H$ is a non-commutative algebraic circuit computing $h$). Additionally, add the axioms of  distributivity of product over addition from \textit{both} left and right (this is needed because we do not have anymore the commutativity axiom in our system).

Since, for $d>2$, the set of generators for the identities over \matd\ are still not well understood, we shall give an explicit formulation only of the system   $\PMtwo$, following the basis of identities of \mattwo\ found by  Drensky \cite{Dre81}.
\begin{definition}[The system $\PMtwo$: proofs of  identities over $\mattwo$]
$\PMtwo$ is the arithmetic proof system whose set of  axioms consists of the following equations (ranging over non-commutative arithmetic circuits):
\begin{align*}
&Addition\: commutativity: \qquad f+g=g+f\\
&Associativity:\qquad   f+(g+h)=(f+g)+f \qquad f\cdot(g\cdot h)=(f\cdot g)\cdot h\\
&Distributivity:\qquad f\cdot(g+h)=f\cdot g+f\cdot h\\
&~~~~~~~~~~~~~~~~~~~~\qquad (g+h)\cdot f= g\cd\ f +h\cdot f\\
&Zero \:element: \qquad f+0=f \qquad f\cdot 0=0\\
&Unit \:element:\qquad f\cdot 1 =f \\
&Genertators:  \qquad S_4(x,y,z,w)=0  \qquad  [[x,y]^2,z]=0\\
&Field \:identities: \qquad c=a+b \qquad d= a'\cdot b'\\
\text{where in } &\text{the Field identities $a,a',b,b',c,d\in \F$, such that the equations hold in $\F$.} \\
&Circuit \ axiom:\qquad F=F'
\text{~~~~if $F$ and $F'$ are (syntactically) identical when} \\
&~~~~~~~~~~~~~~~~~~~~~~~~~~~~~~~~~~~~~~~~~~\text{both are un-winded into \emph{formulas}.}
\end{align*}
\end{definition}

\bigskip

\iddofix{
\\
Open problems
\\
\textbf{Are there small circuits with high Q measure?}
We don't know the answer. Raz' paper might give a way to show that the s-polynomials do not have small circuits.
\\
\textbf{Can we prove that $\PC(\F)$ proof-size of $\llbracket f=0 \rrbracket$ is proportional
to $\PC(\F)$ proof-size of \(f=0\)?}
Think of the $f=s_d$. It seems that
for these formulas it is true that $\PC(\F)$ proof-size of $\llbracket f=0 \rrbracket$
is proportional to $\PC(\F)$ proof-size of \(f=0\). Check. \\ \textbf{Fu: }We checked it a bit; note that XYZW will turn with \(\llbracket \cd \rrbracket \) to a noncommutative polynomial  that respects the order of multiplication (e.g., wi always appear after $x_j$). But there might be some cancelations, so it might be possible that the statement does hold. But checking a bit we concluded that the statement is probably not true.
 \\
 If indeed this is true, then it will show that any identity of \matd\ has polynomial-size $\PC(\F)$ proofs because every \(f=0\) where \(f\) is a matrix identity can be proved efficiently in $\PC(\F)$ (apparently) by simulating Raz and Shpilka 2005 PIT for noncommutative formulas. But there's a problem here: $\PC(\F)$ proofs
 operate with circuits and not formulas and so we can't guarantee a small proof for circuits.
}

\iddofix{
Related computational problems:
\\
\textbf{Given a circuit $F$ such that $f:=\widehat F$ can we determine if $f$ is a matrix identity over \matd? }Note that even if we have a PIT for noncommutative \textit{circuits }(and not merely formulas) it does not solve the problem: because e.g., an $s_4$ identity is \textit{not} an identity as a noncommutative polynomial; and so it would not recognize it. That is, the identities of \mattwo\ have a lot of cancelations possible (almost like the commutative polynomials). In fact, if we have PIT for matrix polynomials of \emph{any} dimension $d\times d$, we can have a PIT for commutative polynomials: just take matrices of dimension 1. So we see that matrix identities are ``much more commutative'' than noncommutative polynomials. On the other hand, recall that we want to count only the number of commutativity axioms needed, and so our proofs might be shorter than the nondeterministic (i.e., witness) complexity of PIT for matrix identities.
\\
\textbf{What is the circuit-size of the hard family of the  s-polynomials?} (when \(d\) is a constant, and \(n\) is the number of variables tend to infinity):
\\
Write a sub-section of open questions:
~~~~~~\\~~~~A. \textbf{Is there a family of identities of small size and high Q measure.
}\\
\textbf{Comment:} Raz 2011 showed that from a high rank tensor one can construct a polynomial that requires big algebraic circuits. Although we have a connection between tensor-rank and the S-polynomials high Q measure we cannot show that ... we might speculate that there are...do we have small circ
\\
\textbf{Comment:}
For any d and any polynomial \(f\) of degree \(d\), we have \(Q(f)=O(n^{d-1})\), because we can write it as a sum of monomials, and arrange it as products
of \(x_1,x_2,...\)
}

\iddofix{Note:don't forget to include the notice that for a lower bound on $Q_G(f_i)$ to hold the family  $f_i$ is not a substitution instance of a finite sets of circuits.}

\section*{Acknowledgements}
We wish to thank V.~Arvind, Albert Atserias, Michael Forbes, Emil Je\v rabek, Kristoffer Arnsfelt Hansen, Jan Kraj\'{i}\v{c}ek, Satya Lokam, Periklis Papakonstantinou, Youming Qiao,  Ran Raz and Amir Shpilka for useful discussions related to this work. We are also greatly indebted to Vesselin Drensky for his help with the bibliography and providing us with his monograph.

\bibliographystyle{plain}
\bibliography{PrfCmplx-Bakoma}
\listoftodos

\end{document}